\pdfoutput=1
\documentclass[11pt,a4paper]{article}
\usepackage[T1]{fontenc}
\usepackage[utf8]{inputenc}
\usepackage{lmodern}
\usepackage[margin=25mm]{geometry}
\usepackage[font=small,labelfont=bf]{caption}

\usepackage{microtype}
\usepackage{mathtools,amssymb,amsthm,bm}
\usepackage{booktabs}
\usepackage{graphicx}
\usepackage{placeins}
\usepackage{needspace}
\usepackage{tikz}
\usepackage{algorithm}
\usepackage{algpseudocode}
\usepackage{array}
\usepackage{enumitem}
\usepackage{hyperref}
\usetikzlibrary{arrows.meta,positioning,calc,patterns}

\setlist[itemize]{topsep=3pt,itemsep=2pt,leftmargin=1.35em}
\setlist[enumerate]{topsep=3pt,itemsep=2pt,leftmargin=1.45em}
\hypersetup{%
  pdftitle={Accuracy, Decision Risk, and Sampling Cost in Quantum Error Mitigation},
  pdfauthor={Vicenzo Scavino Alfaro},
  pdfcreator={LaTeX},
  hidelinks,
  pdfdisplaydoctitle=true}

\theoremstyle{plain}
\newtheorem{theorem}{Theorem}[section]
\newtheorem{proposition}[theorem]{Proposition}

\newtheorem{corollary}[theorem]{Corollary}
\theoremstyle{definition}

\newtheorem{example}[theorem]{Example}

\newcommand{\E}{\mathbb{E}}
\newcommand{\Prob}{\mathbb{P}}
\newcommand{\Var}{\operatorname{Var}}
\newcommand{\Cov}{\operatorname{Cov}}
\newcommand{\tr}{\operatorname{tr}}
\newcommand{\rank}{\operatorname{rank}}
\newcommand{\diag}{\operatorname{diag}}

\newcommand{\calX}{\mathcal{X}}
\newcommand{\calM}{\mathcal{M}}

\newcommand{\calC}{\mathcal{C}}

\newcommand{\calR}{\mathcal{R}}
\newcommand{\calS}{\mathcal{S}}
\newcommand{\calF}{\mathcal{F}}
\newcommand{\calN}{\mathcal{N}}
\newcommand{\MSE}{\mathrm{MSE}}

\newcommand{\Raw}{\mathrm{Raw}}
\newcommand{\ZNE}{\mathrm{ZNE}}
\newcommand{\CDR}{\mathrm{CDR}}
\newcommand{\PEC}{\mathrm{PEC}}

\newcommand{\argmin}{\operatorname*{arg\,min}}

\newcommand{\DRI}{\mathrm{DRI}}
\newcommand{\GRI}{\mathrm{GRI}}
\newcommand{\linspan}{\operatorname{span}}

\newcommand{\cl}{\operatorname{cl}}
\newcommand{\interior}{\operatorname{int}}

\newcommand{\norm}[1]{\left\lVert #1\right\rVert}
\newcommand{\R}{\mathbb{R}}
\newcommand{\1}{\mathbf{1}}

\newcommand{\op}{\mathrm{op}}

\newcommand{\KL}{\operatorname{KL}}
\newcommand{\kl}{\operatorname{kl}}
\newcommand{\TV}{\operatorname{TV}}

\title{Accuracy, Decision Risk, and Sampling Cost in Quantum Error Mitigation}
\author{Vicenzo Scavino Alfaro\\[3pt]
  \small Independent Researcher\\
  \small ORCID: \href{https://orcid.org/0009-0000-2472-9785}{0009-0000-2472-9785}\\
  \small \href{mailto:vicenzoscavino@hotmail.com}{\texttt{vicenzoscavino@hotmail.com}}}
\date{October 23, 2026}

\begin{document}
\maketitle

\begin{abstract}
Quantum error mitigation estimates expectation values that are often used to select or rank candidates. We study how estimation accuracy, decision risk and sampling cost are related in these tasks. The joint law of residual differences determines the risk of a shift-invariant decision; in Gaussian regimes, effective margins and their covariance provide a computable representation. This connects mitigation maps, shared measurement records and shot allocation, and yields an extrapolation design criterion under certified margins.

Finite-shot QAOA-MaxCut controls illustrate two forms of this relationship: a positive affine CDR correction reduces mean-squared error while preserving rankings, whereas signed PEC sampling reduces mean-squared error and increases selection failures at equal evaluation cost. A retrospective, oracle-assisted study of 48 landscapes compares Gaussian-risk and MSE-based method selection using separate fitting, selection and evaluation samples. At the primary total shot budget, with learning amortized over 20 decisions per landscape, failure rates are $0.3083$ for Gaussian-risk selection, $0.4719$ for MSE selection and $0.2333$ for unmitigated estimation. These comparisons describe how uncertainty in candidate differences and the allocation of measurements between learning and evaluation jointly enter the assessment of mitigation.

\end{abstract}

\section{Introduction}

Quantum error mitigation estimates ideal expectation values from noisy quantum records \cite{Cai2023}. Many workflows then use those values to choose an optimizer step, select or rank circuits, retain a top-$k$ set, or test an improvement threshold. Mean-squared error measures the quality of the estimates, while decision risk measures how often their use produces an incorrect choice. We examine how mitigation changes both quantities and how its acquisition cost enters a comparison at a fixed measurement budget.

For a fixed ideal landscape, ranking and argmin depend on differences between candidates. Common shifts contribute to estimation error while leaving these decisions unchanged. For residual field $E_m=a_m+\xi_m$ and contrast map $L\1=0$, the joint gap law $\mathcal L(LE_m)$ describes the relevant uncertainty; in a Gaussian model it is specified by the gap mean and covariance. This representation follows the effects of bias and correlated fluctuations from the estimated values to the decision made from them.

\paragraph{A minimal example.}
Consider three candidates with ideal gaps $\Delta=(1,1)^\top$ and reference contrast $L=\bigl(\begin{smallmatrix}-1&1&0\\-1&0&1\end{smallmatrix}\bigr)$. Two unbiased Gaussian estimators have actual residual covariances
\[
K_A=0.40\,\1\1^\top+0.01I_3,\qquad K_B=0.08I_3.
\]
Their per-candidate MSEs are $0.41$ and $0.08$, respectively, but their gap covariances are
\[
\Sigma_A=\begin{pmatrix}0.02&0.01\\0.01&0.02\end{pmatrix},
\qquad
\Sigma_B=\begin{pmatrix}0.16&0.08\\0.08&0.16\end{pmatrix}.
\]
For $G_m\sim\calN(\Delta,\Sigma_m)$, direct Gaussian integration gives $\Prob[\min_iG_{A,i}<0]=1.54\times10^{-12}$ and $\Prob[\min_iG_{B,i}<0]=1.17\times10^{-2}$. The common-mode contribution to $K_A$ cancels in gaps: B has lower MSE but much higher argmin failure. These covariances describe this single example at its fixed budget.

\paragraph{Contribution and roadmap.}
We develop a QEM-specific description of decision risk and sampling cost through the residual gap law. Mitigation weights, learned maps and shared measurement records determine its mean and covariance; under a Gaussian model these quantities give a calculable risk. The representation supports extrapolation design under certified margins, method selection by estimated Gaussian risk, and bounds on the information needed for decision certification. Controlled CDR and PEC examples describe how changes in estimation error affect candidate selection. The static study then compares risk-based selection, MSE-based selection and fixed methods with the shots used for fitting and selection included. These parts connect the statistical description of a decision to the resources used to make it.

Sections~\ref{sec:decision_gaps}--\ref{sec:gaussian_risk} develop the representation and Gaussian risk, followed by design and selection in Sections~\ref{sec:qem_design}--\ref{sec:estimating_decision_kernel}, certification limits in Section~\ref{sec:decision_limits}, and experiments in Section~\ref{sec:numerical_proof_of_mechanism}. Figure~\ref{fig:pipeline} connects the measurement records, estimation and decision objectives, and shot accounting.

\begin{figure}[!ht]
    \centering
    \includegraphics[width=\linewidth]{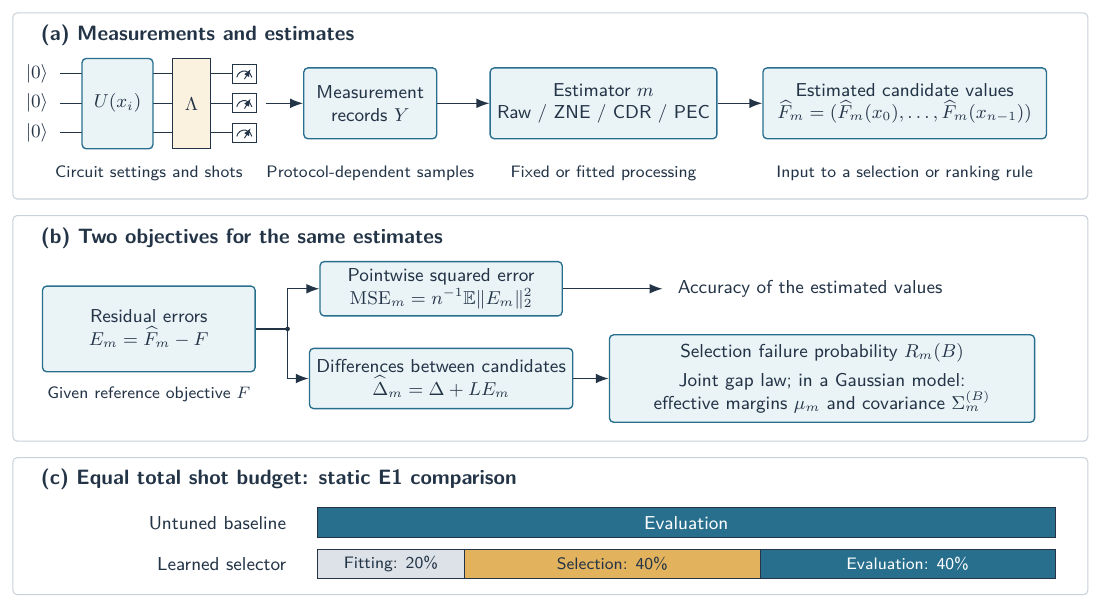}
    \caption{Estimation, candidate selection and sampling cost in quantum error mitigation.
(a) A schematic noisy circuit supplies protocol-dependent measurement records from which an estimator produces candidate values.
(b) Given an ideal reference objective $F$, MSE measures pointwise squared estimation error, while the joint law of the estimated gaps determines selection risk. In a Gaussian model, this law is specified by effective margins and their covariance.
(c) Normalized shot allocations in the full-cost static E1 comparison (Section~\ref{sec:e1_reanalysis}): fitting, method selection and evaluation use $30nt$, $60nt$ and $60nt$ shots across 20 decisions, where $n=25$ is the number of candidates and $t$ the primitive prefix length. Untuned baselines devote the same $150nt$ total to evaluation. Ideal references are simulation inputs in E1; their acquisition and classical computation are outside this shot budget.}
    \label{fig:pipeline}
\end{figure}

\section{Relation to prior work}

\paragraph{Prediction and decision objectives.}
Predict-then-optimize methods distinguish prediction error from the loss induced by using a prediction in an optimization problem. The Smart Predict-then-Optimize framework makes that distinction explicit through a decision loss \cite{ElmachtoubGrigas2017}. Here the estimates come from quantum measurement records: mitigation changes their bias, dependence and sampling cost. The residual gap representation describes how those changes affect selection among a fixed set of candidates.

\paragraph{QEM estimators and uncertainty.}
ZNE and PEC use noise extrapolation and quasiprobability sampling to estimate ideal observables \cite{Temme2017,LiBenjamin2017,Endo2018}; demonstrations and software establish their practical setting \cite{Kandala2019,LaRose2022}. CDR, vnCDR and UNITED learn corrections from auxiliary circuits or variable-noise records \cite{Czarnik2021,Strikis2021,Lowe2021,Bultrini2023}. Other primitives include virtual distillation, readout correction and symmetry verification \cite{Huggins2021,Bravyi2021,BonetMonroig2018}; Ref.~\cite{Cai2023} reviews their assumptions and costs. Model-violation and robust-design analyses address uncertainty in mitigated observables \cite{Govia2025,Prodius2026}. The present work follows that residual uncertainty through a decision using several estimated values.

\paragraph{Finite-shot benchmarking and limits.}
Demarty et al. include QEM sampling overhead in testing whether one estimated energy lies in a certified interval \cite{Demarty2026}. Here the target is a joint landscape decision, governed by coupled gaps. Budget-dependent changes in mitigation performance have also been observed by Bultrini et al. \cite{Bultrini2023}. The finite-shot MSE boundary for Richardson ZNE describes the bias--variance balance for expectation-value estimation at fixed cost \cite{Scavino2026Boundary}. The present analysis follows the joint estimation errors into candidate decisions through their margins, covariance and allocation. General QEM estimation-cost limits use distinguishability and information geometry \cite{Takagi2022,Takagi2023,Tsubouchi2022,Quek2024}. Our supporting converse bounds errors between instances where the incumbent loses strict optimality, for a fixed primitive measurement family and allocation.

\paragraph{Correlated selection and statistical experiments.}
Best-arm identification and ranking-and-selection already study decision errors and instance-dependent costs \cite{Bechhofer1954,GuptaPanchapakesan1979,GarivierKaufmann2016}, including correlated observations and covariance-adaptive sampling \cite{GuptaJoshiYagan2021,Saad2023}. Linear best-arm identification supplies fixed-confidence guarantees \cite{JedraProutiere2020}; universal instance-specific optimality can fail in fixed-budget settings \cite{Degenne2023}. Our contribution is the QEM instantiation: mitigation weights, learned maps, shared records and sampling overhead determine the gap law and its acquisition cost. Blackwell comparison concerns entire parameterized experiments \cite{Blackwell1953}; its explicit Gaussian form is separated from fixed-rule risk comparisons in Appendix~\ref{app:gaussian_comparison}.

The analysis uses Gaussian comparison, anti-concentration, covariance concentration and large-deviation results \cite{Slepian1962,Sidak1967,LedouxTalagrand1991,Chernozhukov2013,Boucheron2013,Vershynin2018,KoltchinskiiLounici2017,DemboZeitouni1998}. Their respective approximation regimes are specified with the results below.

\section{Residual gaps and decision sufficiency}
\label{sec:decision_gaps}

Let $\calX=\{x_0,\ldots,x_{n-1}\}$ be a finite candidate set and $F$ its ideal objective, written for minimization. Assume a unique reference optimum $x_0$, with gaps $\Delta_i=F(x_i)-F(x_0)>0$. A method $m$ produces
\[
\widehat F_m^{(B)}=F+E_m^{(B)},\qquad E_m^{(B)}=a_m+\xi_m^{(B)},\qquad \E\xi_m^{(B)}=0.
\]
Residuals are square-integrable. Use $a_m^{(B)}$ whenever the residual mean depends on budget, including through nonlinear post-processing or budget-dependent maps. Write $K_m^{(B)}=\Cov(\xi_m^{(B)})$. Then
\begin{equation}\label{eq:mse_v6}
\MSE_m(B)=\frac{\|a_m\|_2^2+\tr K_m^{(B)}}{n}.
\end{equation}
The reference contrast $L:\R^n\to\R^d$, with $d=n-1$ and $(Lv)_i=v_i-v_0$, removes common shifts. The estimated gaps have effective mean and covariance
\begin{equation}\label{eq:gap_model}
L\widehat F_m^{(B)}=\mu_m+L\xi_m^{(B)},\qquad
\mu_m=\Delta+La_m,\qquad \Sigma_m^{(B)}=LK_m^{(B)}L^\top.
\end{equation}
Counting ties against $x_0$, the argmin risk is
\begin{equation}\label{eq:argmin_risk}
R_m(B)=\Prob[\exists i:\ \mu_{m,i}+(L\xi_m^{(B)})_i\le0].
\end{equation}
Other tie conventions change the boundary event and must be specified.

More generally, a finite gap decision has a contrast map $L_D\1=0$, an ideal deterministic margin $b_D$, and a Borel success set $\calS$. It succeeds when $b_D+L_DE_m\in\calS$. Mitigation bias remains part of $E_m$.

\begin{theorem}[Gap representation and sufficiency]
\label{thm:quotient_factorization}\label{thm:minimal_gap_law}\label{thm:gap_law_iff}
Let $\ker L=\linspan\{\1\}$. Every measurable finite-action rule satisfying $\delta(v+c\1)=\delta(v)$ factors uniquely as $\delta=\bar\delta\circ L$, with $\bar\delta$ measurable on $L(\R^n)$.

For any fixed contrast $L_D$, the law of $L_DE_m$ determines all gap-decision risks. Conversely, probabilities of all shifted orthants, or all halfspaces $h^\top L_DE_m\le t$, determine that law.
\end{theorem}

\begin{proof}
Equal contrasts differ by a common shift, so $\bar\delta$ is well defined; a linear right inverse $W$ on $L(\R^n)$ gives $\bar\delta=\delta\circ W$ and its measurability. A gap risk is the probability of $\calS^c-b_D$. Shifted orthants determine the multivariate distribution, including boundary atoms; halfspaces determine all one-dimensional projections and hence the law by Cram\'er--Wold.
\end{proof}

This identification uses all thresholds in a determining family; one fixed decision risk need not identify the law. The conclusion is uniqueness of the law from event probabilities; measurable reconstruction from a summary requires a specified measurable structure.

Adding any square-integrable random $C\1$ leaves $LE_m$ unchanged, even if $C$ is correlated with $E_m$. An independent centered $C$ increases MSE by $\Var(C)$ without changing a gap decision. Common affine transformations with positive slope also preserve rankings (Appendix~\ref{app:affine_controls}).

\section{Marginal accuracy and joint decision risk}
\label{sec:marginal_obstruction}

MSE averages the squared estimation errors across candidates. More generally, a marginal benchmark describes the individual residual laws through biases, variances, confidence intervals or other separable expected losses. Decision risk also depends on how these errors occur jointly, as the following Gaussian construction illustrates.

\begin{theorem}[Marginal benchmarks can reverse decision ordering]
\label{thm:nogo_marginal}
Among finite Gaussian residual fields, identical pointwise marginals can have different argmin risks. A deterministic method selector using only those marginals therefore cannot always select the lower-risk method; a randomized selector has error probability at least $1/2$ on one of two indistinguishable labelings. Moreover, lower MSE can coexist with higher risk and, under $B^{-1/2}$ shot scaling, a strictly smaller Gaussian failure exponent.
\end{theorem}

\begin{proof}
For two candidates with ideal margin $\Delta>0$, take centered Gaussian residuals with
\[
K_A=\sigma^2I_2,\qquad
K_B=\sigma^2\begin{pmatrix}1&\rho\\\rho&1\end{pmatrix},\qquad \sigma>0,\quad0<\rho<1.
\]
Both pointwise marginals are $\calN(0,\sigma^2)$, but their gap variances are $2\sigma^2$ and $2(1-\rho)\sigma^2$, so
\begin{equation}
R_A=\Phi\!\left(-\frac{\Delta}{\sqrt2\sigma}\right)
>\Phi\!\left(-\frac{\Delta}{\sqrt{2(1-\rho)}\sigma}\right)=R_B.
\end{equation}
Swapping method labels preserves every marginal benchmark and reverses the better label. One labeling must defeat any deterministic selector; the two error probabilities of a randomized selector sum to one.

Adding an independent centered common-mode variable $C\1$ with variance $\tau^2>0$ to B leaves its gap unchanged and raises its MSE from $\sigma^2$ to $\sigma^2+\tau^2$. Thus A has lower MSE and higher risk. Scaling all residuals, including $C$, by $B^{-1/2}$ preserves that MSE ordering and gives
\begin{equation}
I_A=\frac{\Delta^2}{4\sigma^2},\qquad
I_B=\frac{\Delta^2}{4(1-\rho)\sigma^2},\qquad
\lim_{B\to\infty}\frac1B\log\frac{R_A(B)}{R_B(B)}=I_B-I_A>0.
\end{equation}
The last equality follows from the one-dimensional Gaussian tail rates.
\end{proof}

This obstruction concerns ambient marginals. Equal one-dimensional \emph{gap} laws can also leave joint risk unidentified; Appendix~\ref{app:joint_witnesses} gives that distinct witness.

MSE can identify decision quality in a restricted family when it determines all relevant margins and gap uncertainty. For general residual fields, decision analysis also requires the dependence induced by the measurement protocol.

\section{From primitive measurements to the decision kernel}
\label{sec:primitive_kernel}

To connect the gap covariance to the measurement protocol, stack its primitive sample means into $Y(B)$, with actual covariance $C_{\rm dev}(B)=\Cov(Y(B))$. Its coordinates may include noise scales, training circuits, quasiprobability branches and readout-calibration records. For a fixed linear estimator $\widehat F_m=M_mY(B)$,
\begin{equation}\label{eq:qem_pullback_finite}
K_m^{(B)}=M_mC_{\rm dev}(B)M_m^\top,\qquad
\Sigma_m^{(B)}=(LM_m)C_{\rm dev}(B)(LM_m)^\top.
\end{equation}
\paragraph{Budget convention.}
In formulas displaying $B$, unadorned $K_m$ and $\Sigma_m$ denote the limits of $BK_m^{(B)}$ and $B\Sigma_m^{(B)}$, when those limits exist. Under exact shot scaling, $K_m^{(B)}=K_m/B$ and $\Sigma_m^{(B)}=\Sigma_m/B$; under asymptotic scaling these identities have $o(B^{-1})$ remainders. This distinguishes the actual finite-budget covariance from the per-unit kernel used in Gaussian rates. Write $\Sigma_{m,B}:=B\Sigma_m^{(B)}$ for the rescaled finite-budget covariance; it need not equal its limit $\Sigma_m$.

If $BC_{\rm dev}(B)\to C_{\rm dev}$, the per-unit kernel is
\begin{equation}\label{eq:qem_pullback_kernel}
\Sigma_m=LK_mL^\top=(LM_m)C_{\rm dev}(LM_m)^\top.
\end{equation}
The rank is at most $\rank C_{\rm dev}$. A zero limiting gap variance gives only $o(B^{-1})$ finite-budget variance, which need not vanish. The two zero-variance criteria coincide under exact shot scaling. Appendix~\ref{app:finite_budget_degeneracy} gives the scalar identity and deterministic criterion, a shared-record counterexample and a bound for retaining smaller variance terms.

Factorization alone imposes no further restriction when $L$ has full row rank: for any $S\succeq0$, a right inverse $W$ of $L$, $C_{\rm dev}=I$ and $M=WS^{1/2}$ realize $S$. Redundant contrast rows retain their linear dependencies. Physical restrictions must therefore come from the allowed protocol and allocation.

For nonlinear $\widehat F_m=g_m(Y(B))$, replace $M_m$ by a uniformly bounded Jacobian and require a first-order remainder with mean square $o(B^{-1})$. The covariance formula then holds to leading order. Learned maps must either be held fixed conditional on training, or modeled jointly with their training and target records. The precise expansion and the finite-shot central-limit statement appear in Appendix~\ref{app:covariance_clt}.

\subsection{Primitive noise and shared calibration}
\label{subsec:anchoring_cdev}
For a Pauli outcome in $\{-1,+1\}$ at setting $(x,s)$, let
\[
f_{x,s}(p)=\tr[O_x\calN_{p,s}(\rho_x)],\qquad
\Var(\overline Y_{x,s}\mid p)=\frac{1-f_{x,s}(p)^2}{B_{x,s}}.
\]
With $B_{x,s}/B\to\lambda_{x,s}>0$, independent primitive streams give a diagonal per-unit shot covariance with entries $(1-f_{x,s}(\bar p)^2)/\lambda_{x,s}$. Joint measurements or reused records require their full conditional covariance, denoted $D_{\rm shot}$.

Write $p_B=\bar p+\eta_B$ and $Y(B)=f(p_B)+\varepsilon_B$. Suppose $\E\eta_B=0$, $B\Cov(\eta_B)\to V_p$, and the calibration response has the mean-square expansion
\[
f(p_B)=f(\bar p)+J_p\eta_B+r_B,
\qquad \E\|r_B\|_2^2=o(B^{-1}).
\]
If $\E[\varepsilon_B\mid p_B]=0$ and $B\E\Cov(\varepsilon_B\mid p_B)\to D_{\rm shot}$, the per-unit covariance is
\begin{equation}\label{eq:cdev_physical_delta}
C_{\rm dev}=\lim_{B\to\infty}B\Cov(Y(B))
=D_{\rm shot}+J_pV_pJ_p^\top,
\qquad (J_p)_{(x,s),r}=\partial_{p_r}f_{x,s}(\bar p).
\end{equation}
Equation~\eqref{eq:qem_pullback_kernel} projects this covariance into gap space through $LM_m$; any nonzero remainder mean belongs to the bias. Appendix~\ref{app:calibration_covariance} gives sufficient moment conditions, the correction for reused calibration data, and the treatment of persistent drift.

\section{Gaussian risk and decision rates}
\label{sec:gaussian_risk}

To turn the propagated covariance into a computable risk, use the working Gaussian model
\[
L\widehat F_m(B)=\mu_m+B^{-1/2}Z_m,\qquad
Z_m\sim\calN(0,\Sigma_m),\qquad \mu_m=\Delta+La_m.
\]
A central-limit justification for bounded primitive outcomes is given in Appendix~\ref{app:normal_approximation}. A CLT controls central, absolute-probability approximations; it does not preserve probabilities of order $\exp(-BI)$. The rates below are exact for this Gaussian surrogate. Matching them to the original shot model requires an additional large- or moderate-deviation argument.

\Needspace{10\baselineskip}
\begin{proposition}[Gaussian risk and argmin rate]
\label{prop:gaussian_risk}\label{thm:decision_exponent}
If the gap marginals are nondegenerate, the strict-argmin risk is
\begin{equation}\label{eq:gaussian_risk}
R_m(B)=1-\Phi_d(\sqrt B\,\mu_m;\Sigma_m),
\end{equation}
where $\Phi_d$ is the centered Gaussian lower-orthant distribution function. For a general success set $\calS$, use $1-\Prob[\mu_m+B^{-1/2}Z_m\in\calS]$ instead.
If all effective margins and relevant gap variances are positive, the argmin rate is
\begin{equation}\label{eq:argmin_exponent}
\lim_{B\to\infty}-B^{-1}\log R_m(B)
=I_m=\min_i\frac{\mu_{m,i}^2}{2\Sigma_m(i,i)}.
\end{equation}
Within this exact Gaussian model, a zero-variance coordinate with positive margin never causes failure. A nonpositive margin gives no positive rare-event exponent under the tie-as-failure convention.
\end{proposition}

\begin{proof}
Gaussian symmetry turns success into the lower-orthant probability in Eq.~\eqref{eq:gaussian_risk}. For positive margins, failure is a finite union of one-dimensional Gaussian lower-tail events. Its probability lies between the largest marginal tail and the sum of those tails, so its exponential rate is the smallest marginal rate. Degenerate coordinates in this Gaussian model are deterministic and are treated separately. The extension to general regular failure regions is in Appendix~\ref{app:gaussian_proofs}.
\end{proof}

Even this diagonal \emph{gap} rate uses ambient correlations, since
\[
\Sigma_m(i,i)=K_m(i,i)+K_m(0,0)-2K_m(i,0).
\]
Off-diagonal gap covariances affect finite-budget orthant risk. For rates, the failure event determines which contrasts matter: argmin failure is a union of single crossings, and exact top-$k$ set failure is a union of pairwise crossings (Appendix~\ref{app:topk_rules}). Correlations in the reference-gap kernel can enter the latter rate through the variances of those pairwise contrasts. Simultaneous inversion patterns and tolerant-failure events can require several active constraints. Theorem~\ref{thm:diagonal_offdiagonal} in Appendix~\ref{app:gaussian_proofs} gives their quadratic cost, followed by prefactors and covariance sensitivity.

\section{Decision-aware QEM design}
\label{sec:qem_design}

Equations~\eqref{eq:qem_pullback_kernel} and \eqref{eq:argmin_exponent} connect the effective margins and covariance of each mitigation map to its Gaussian decision rate. ZNE permits direct coefficient design; CDR and PEC illustrate how training uncertainty and sampling overhead enter that rate. Primitive shot costs enter through the allocation-dependent covariance.

\subsection{ZNE: bias cancellation and noise amplification}
\label{sec:zne_signature}
For $q$ noise scales, let $\widehat F_\ell=F+b_\ell+\xi_\ell$ estimate the landscape at $\lambda_\ell$, $\ell=1,\ldots,q$; digital folding supplies one implementation \cite{GiurgicaTiron2020}. A local expansion $F_\lambda(x)=F(x)+\sum_{r=1}^p\beta_r(x)\lambda^r+O(\lambda^{p+1})$, uniform over the extrapolation window, motivates Richardson coefficients satisfying
\[
\sum_\ell c_\ell=1,\qquad \sum_\ell c_\ell\lambda_\ell^r=0\quad(r=1,\ldots,p).
\]
Assume these constraints are feasible. For $\widehat F_{\ZNE}=\sum_\ell c_\ell\widehat F_\ell$, linearity gives
\begin{equation}
a_{\ZNE}=\sum_\ell c_\ell b_\ell,\qquad
\Sigma_{\ZNE}^{(B)}=L\left(\sum_{\ell,r}c_\ell c_r K_{\ell r}^{(B)}\right)L^\top,
\end{equation}
where $K_{\ell r}^{(B)}=\Cov(\xi_\ell,\xi_r)$. The per-unit kernel replaces these cross-covariances by $K_{\ell r}=\lim_B BK_{\ell r}^{(B)}$, when the limits exist. For independent scales with allocations $B_\ell=w_\ell B$, $w_\ell>0$, and per-shot gap kernels $\Sigma_\ell$,
\begin{equation}
\Sigma_{\ZNE}=\sum_\ell \frac{c_\ell^2}{w_\ell}\Sigma_\ell.
\end{equation}
Under exact shot scaling the actual covariance is $\Sigma_{\ZNE}/B$. Equal allocations and identical scale kernels give amplification $q\|c\|_2^2\Sigma_0$. Thus bias cancellation can increase the margin while amplification lowers the rate in Eq.~\eqref{eq:argmin_exponent}. For one active gap with positive margins, ZNE improves the rate precisely when its $\mu_i^2/\Sigma_{ii}$ exceeds Raw's. This criterion differs from an MSE comparison; finite-budget ordering additionally depends on the joint tails.

\subsection{Decision-optimal ZNE coefficients under certified margins}

Write the feasible Richardson constraints as $Ac=\rho_{\rm R}$. Any free coefficients can then be chosen for the decision task.

For a chosen critical band $\calC\subseteq\{1,\ldots,d\}$, define the per-unit per-scale or cross-scale quadratic form
\begin{equation}
    Q_i=\left[ e_i^\top L K_{\ell r}L^\top e_i\right]_{\ell,r=1}^q,
    \qquad i\in\calC,
\end{equation}
so that the per-unit variance of critical gap $i$ under ZNE coefficients $c$ is $c^\top Q_i c$.

\begin{theorem}[Decision-optimal ZNE under certified margins]
\label{thm:decision_optimal_zne}
For each critical gap, write the ZNE effective margin as
\begin{equation}
    \mu_i(c)=\Delta_i+\sum_\ell c_\ell b_{\ell,i}
\end{equation}
where $b_{\ell,i}:=(Lb_\ell)_i$ is the bias contribution in gap space; for the reference contrast this is $b_\ell(x_i)-b_\ell(x_0)$.  Let $\calC$ be nonempty. Assume fixed certified bounds $\underline\mu_i>0$ such that $\mu_i(c)\ge\underline\mu_i$ for every $c$ satisfying $Ac=\rho_{\rm R}$ and every $i\in\calC$. A known constant positive margin is the special case $\underline\mu_i=\mu_i$. Under these conditions, the coefficients that maximize the certified worst critical-gap decision signal-to-noise ratio solve the second-order cone program
\begin{equation}
\label{eq:decision_opt_zne_socp}
\begin{aligned}
    \min_{c,u}\quad & u \\
    \text{subject to}\quad
        & A c=\rho_{\mathrm R},\\
        & \norm{Q_i^{1/2}c}_2\le u\,\underline\mu_i,
        \qquad i\in\calC .
\end{aligned}
\end{equation}
Equivalently, it minimizes
\begin{equation}
    \max_{i\in\calC}\frac{c^\top Q_i c}{\underline\mu_i^2}.
\end{equation}
These coefficients generally do not coincide with coefficients that minimize ambient MSE or the average pointwise variance.
\end{theorem}

\begin{proof}
For fixed certified margins, the certified worst critical-gap exponent is monotone in $\min_i \underline\mu_i^2/(c^\top Q_i c)$.  Maximizing this quantity is equivalent to minimizing the maximum displayed ratio.  Since each $Q_i$ is positive semidefinite as a covariance quadratic form, the constraint $\norm{Q_i^{1/2}c}_2\le u\underline\mu_i$ is a second-order cone constraint, while the Richardson constraints are affine.
\end{proof}

The uniform lower bound over the entire affine family $Ac=\rho_{\rm R}$ forces each affine margin $\mu_i(c)$ to be constant along $\ker A$: otherwise, one sign of a null direction makes it arbitrarily negative. Its value need not be known exactly if a positive lower bound is certified. Nonconstant margins require a restricted coefficient domain, included in the optimization, on which the positive bounds are certified; a convex domain preserves convexity. A first-order approximation alone is insufficient without uniform control of its remainder.

For independent scale estimators and fixed certified margins, Appendix~\ref{app:joint_allocation} gives a convex continuous relaxation for joint coefficients and shot allocation.

\subsection{CDR: training uncertainty in gap space}
\label{sec:cdr_signature}
For a linearized CDR map, write
\[
\widehat F_{\CDR}=X_{\rm feat}(\theta+\delta\theta^{(B)})+\varepsilon_{\rm tar},
\qquad a_{\CDR}=X_{\rm feat}\theta-F,
\]
where row $i$ of $X_{\rm feat}$ is $\phi(x_i)^\top$ and the two random residual terms are centered. Any nonzero means are absorbed into $a_{\CDR}$. Let $V_\theta^{(B)}=\Cov(\delta\theta^{(B)})$ and $K_{\rm tar}^{(B)}=\Cov(\varepsilon_{\rm tar})$. Then
\begin{equation}
\Sigma_{\CDR}^{(B)}=(LX_{\rm feat})V_\theta^{(B)}(LX_{\rm feat})^\top
+LK_{\rm tar}^{(B)}L^\top+\Sigma_{\rm cross}^{(B)},
\end{equation}
with both training--target cross terms included when the records share randomness. The training contribution to entry $(i,j)$ is $(\phi_i-\phi_0)^\top V_\theta^{(B)}(\phi_j-\phi_0)$. Small feature differences can therefore suppress training uncertainty in gaps even when pointwise uncertainty is large. Deterministic transfer error still changes $\mu_{\CDR}=\Delta+La_{\CDR}$.

For a nonlinear fitted procedure this is the leading covariance of the full training-plus-target delta expansion. Conditional analysis instead holds the learned map fixed and treats its conditional target covariance and bias; it does not add unconditional training variance a second time. Multiplication by $B$ and a shot-scaling limit give the per-unit version.

\subsection{PEC: overhead and decision risk}
\label{sec:pec_signature}
PEC uses signed quasiprobability sampling to cancel a characterized noise model \cite{vandenBerg2023}. In the special unbiased scalar-overhead model, $K_{\PEC}=\Gamma^2K_{\rm idealMC}$ with $\Gamma\ge1$, hence
\begin{equation}
\mu_{\PEC}=\Delta,\qquad
\Sigma_{\PEC}=\Gamma^2LK_{\rm idealMC}L^\top,
\qquad I_{\PEC}=\Gamma^{-2}I_{\rm unbiased}.
\end{equation}
The last comparison holds for the same covariance shape and positive margins. Circuit-dependent weights, observable-dependent variance, truncation or noise-model mismatch require the actual PEC covariance and bias; a global scalar factor is then insufficient. The signed-sampling control in Appendix~\ref{app:signed_pec} evaluates this observable-dependent covariance directly.

\paragraph{One-gap PEC--CDR comparison.}
For one fixed gap in the scalar-overhead model, use scalar margins and variances: put $\Sigma_{\rm id}=(LK_{\rm idealMC}L^\top)_{ii}$ and $\mu_{\CDR}=\Delta+\delta_{\CDR}$. With positive variances and $\Delta,\mu_{\CDR}>0$, PEC has the better rate exactly when
\begin{equation}\label{eq:pec_cdr_boundary}
\frac{\Delta^2}{\Gamma^2\Sigma_{\rm id}}
>\frac{(\Delta+\delta_{\CDR})^2}{\Sigma_{\CDR}}.
\end{equation}
A nonpositive CDR margin instead makes the strict decision inconsistent; squaring it would give a misleading comparison. Figure~\ref{fig:pec_cdr_phase} shows the rate boundary and the physical constraint $\Gamma^2\ge1$.

\begin{figure}[t]
    \centering
    \begin{tikzpicture}[>=Latex, x=1.48cm, y=1.05cm]
        \fill[red!12] (-1.55,0) rectangle (-1,3.05);
        \fill[orange!14] (-1,0) rectangle (2.15,3.05);
        \fill[blue!14] (-1,0) rectangle (-0.42,3.05);
        \fill[blue!14, domain=-0.42:2.05, samples=120]
            (-0.42,0) -- plot (\x,{1/(1+\x)^2}) -- (2.05,0) -- cycle;
        \fill[pattern=north east lines, pattern color=gray!24] (-1,0) rectangle (2.05,1);
        \draw[->] (-1.6,0) -- (2.25,0) node[right] {$b=\delta_{\CDR}/\Delta$};
        \draw[->] (-1.55,0) -- (-1.55,3.18) node[above] {$\Gamma^2$};
        \draw[dashed, red!70!black] (-1,0) -- (-1,3.05);
        \node[red!70!black, rotate=90, font=\tiny] at (-0.93,1.52) {$b=-1$};
        \draw[dashed, gray!55!black, line width=0.85pt] (-1,1) -- (2.05,1)
            node[pos=0.78, above=2pt, font=\scriptsize, text=gray!55!black] {physical floor $\Gamma^2\ge1$};
        \draw[blue!65!black, line width=1.1pt, domain=-0.42:2.05, samples=120]
            plot (\x,{1/(1+\x)^2});
        \node[blue!65!black, font=\scriptsize] at (1.03,0.52)
            {$\Gamma^2=\dfrac{1}{(1+b)^2}$};
        \node[align=left, font=\tiny, red!70!black, anchor=west] at (-1.47,2.68)
            {CDR\\inconsistent\\PEC wins};
        \node[align=center, font=\scriptsize, blue!70!black] at (-0.33,0.36)
            {PEC wins\\overhead small};
        \node[align=center, font=\small, orange!80!black] at (1.24,2.33)
            {CDR wins\\overhead dominates};
        \draw (-1,0.06) -- (-1,-0.06) node[below] {$-1$};
        \draw (0,0.06) -- (0,-0.06) node[below] {$0$};
        \draw (1,0.06) -- (1,-0.06) node[below] {$1$};
        \draw (-1.61,1) -- (-1.49,1) node[left=3pt] {$1$};
        \draw (-1.61,2) -- (-1.49,2) node[left=3pt] {$2$};
        \draw (-1.61,3) -- (-1.49,3) node[left=3pt] {$3$};
    \end{tikzpicture}
    \caption{Analytic PEC--CDR boundary from Eq.~\eqref{eq:pec_cdr_boundary}, for $\rho=\Sigma_{\CDR}/\Sigma_{\rm id}=1$. For $b=\delta_{\CDR}/\Delta>-1$, PEC has the better rate below $\Gamma^2=\rho/(1+b)^2$; for $b\le-1$, CDR is inconsistent. The hatched region $\Gamma^2<1$ is physically excluded. This schematic contains no experimental data.}
    \label{fig:pec_cdr_phase}
\end{figure}

\section{Estimating decision risk and selecting a method}
\label{sec:estimating_decision_kernel}

The design above assumes specified margins or certified lower bounds and a covariance model. Selecting a method from data at budget $B$ instead uses estimates of the budget-specific mean and rescaled finite covariance $\Sigma_{m,B}=B\Sigma_m^{(B)}$. Appendix~\ref{app:estimation_stability} gives centered-sample construction, covariance concentration and plug-in stability bounds. These estimates support a Gaussian surrogate only when the gap-law approximation is justified.

\paragraph{Information needed for the reference gaps.}
Estimating noisy means and covariances does not by itself identify the ideal contrasts required by Algorithm~\ref{alg:decision_aware_selection}. In a readout model, two unknown state--channel pairs can yield identical complete target records and opposite ideal minimizers (Example~\ref{ex:readout_identifiability}). Repeating those target measurements cannot distinguish them. Calibration or additional physical information must resolve this decision-relevant ambiguity, with its acquisition cost included in the method comparison.

\paragraph{Choose and control the critical band.}
An ideal band is $\calC_\tau=\{i:0<\Delta_i\le\tau\}$; a practical screen can use small estimated standardized margins $\widehat\mu_{m,i}/\sqrt{\widehat\Sigma_{m,B}(i,i)}$. Choose it using an independent pilot, known simulation truth, or sample splitting; reusing the observations for screening and risk reporting requires accounting for selection. Declare any numerical regularization. Full-landscape risk combines the restricted-band risk with a bound for omitted competitors, such as Eq.~\eqref{eq:omitted_gap_bound}.

Only a known zero \emph{finite-budget} variance makes a gap deterministic: a positive margin can then be omitted, while a nonpositive margin makes strict-argmin failure certain. A zero limiting-kernel entry or sample variance is insufficient; retain its finite-budget uncertainty or bound its omitted risk. An empty band has restricted risk zero, but neither DRI nor the coefficient optimization is invoked.

\paragraph{Summarize risk and select.}
For the Gaussian kernel model on a nonempty band with positive diagonal variances, define the decision reliability index (DRI)
\begin{equation}
\DRI_m^{\rm unit}(\calC)=\min_{i\in\calC}\frac{\mu_{m,i}}{\sqrt{\Sigma_m(i,i)}},
\qquad \DRI_m(B,\calC)=\sqrt B\,\DRI_m^{\rm unit}(\calC),
\end{equation}
and obtain the Gaussian risk index $\GRI_m(B,\calC)$ by restricting Eq.~\eqref{eq:gaussian_risk} to $\calC$. The signed DRI summarizes the worst standardized gap. For fixed positive margins and a fixed kernel, the band's rate is $(\DRI_m^{\rm unit})^2/2$; a nonpositive margin gives rate zero. GRI retains the joint covariance and budget. At finite budget, replace the mean and kernel by their budget-specific counterparts. Algorithm~\ref{alg:decision_aware_selection} selects by estimated finite-budget Gaussian risk, with independent fitting, selection and evaluation data and the shot costs of all three stages included.

\begin{algorithm}[!ht]
\caption{QEM selection for strict argmin}
\label{alg:decision_aware_selection}
\begin{algorithmic}[1]
\Require Candidate methods $m\in\{\Raw,\ZNE,\CDR,\PEC\}$; budget $B$; estimated residual bias $\widehat a_m$; estimated rescaled finite covariance $\widehat K_{m,B}$; reference contrasts $(Lv)_i=v_i-v_0$; estimated ideal landscape $\widehat f$.
\State Use independent fitting data to estimate maps and an independent selection sample to estimate risk; reserve further data for evaluation.
\State Charge fitting and selection shots to the total resource ledger.
\State Compute $\widehat\Delta=L\widehat f$, where $\widehat f$ estimates the ideal landscape, not a biased noisy target already carrying $\widehat a_m$.
\For{each method $m$}
    \State Compute $\widehat\mu_m=\widehat\Delta+L\widehat a_m$.
    \State Compute $\widehat\Sigma_{m,B}=L\widehat K_{m,B}L^\top$; handle certified deterministic gaps separately.
    \State Compute $\widehat R_m(B)=1-\Prob[\widehat\mu_m+B^{-1/2}Z_m>0]$, where $Z_m\sim\calN(0,\widehat\Sigma_{m,B})$ and $>0$ is componentwise; ties count as failure.
    \State Report uncertainty diagnostics for unresolved zero-variance estimates.
\EndFor
\State Select $m^\star\in\argmin_m \widehat R_m(B)$.
\State \Return selected method $m^\star$, its estimated Gaussian surrogate risk $\widehat R_{m^\star}(B)$, and a stability/confidence diagnostic.
\end{algorithmic}
\end{algorithm}

\paragraph{Other decisions.}
For nonunique or near-degenerate optima, one may instead require the empirical minimizer set to lie inside an $\epsilon$-optimal target set. Top-$k$, ranking and step-acceptance events likewise use pairwise contrasts. Applying the selection principle to another decision requires its own mean $b_D+L_Da_m$, covariance $L_DK_m^{(B)}L_D^\top$ and success set $\calS$, rather than the strict positive-orthant event in Algorithm~\ref{alg:decision_aware_selection}; Appendix~\ref{app:topk_rules} states their tie conventions.

\section{Limits of decision certification}
\label{sec:decision_limits}

The preceding risk and selection calculations evaluate a specified rule at one instance. A lower bound for arbitrary rules compares errors across alternative instances. Appendix~\ref{app:converse_proofs} derives that bound from the shot records, for fixed primitive measurements and non-adaptive allocations, including the finite-KL condition needed when allocation fractions vanish.

For Gaussian gaps with covariance $\Sigma/B$, $\Sigma\succ0$, the closest alternative with $m_i\le0$ has exact balanced binary-testing exponent $\mu_i^2/(8\Sigma_{ii})$; the pointwise plug-in rate is $\mu_i^2/(2\Sigma_{ii})$. In one dimension, the plug-in threshold zero has error $1/2$ at the boundary alternative $m=0$, while the balanced test uses threshold $\mu/2$. Thus the two exponents answer different decision questions.

Fisher information connects this comparison to the local geometry of the measurement model. Appendix~\ref{app:local_fisher} first gives an exact quadratic tangent benchmark, then states the uniform regularity and localization conditions required to identify it with a physical local KL minimum. A depolarizing-qubit example realizes this local benchmark within a specified physical family.

\section{Numerical experiments}
\label{sec:numerical_proof_of_mechanism}

The numerical studies connect three aspects of the analysis: changes in estimation and decision quality, prediction of decision risk, and method selection at full shot cost. Noise-model controls and a hardware micro-cell examine the role of the measurement assumptions. QAOA-MaxCut instances \cite{Farhi2014} are evaluated with Qiskit Aer density-matrix simulation \cite{Qiskit2024} and multinomial finite-shot sampling.

The static, dynamic and covariance studies use archived records. Their instance sets and initial endpoints were specified in advance \cite{Nosek2018}. The sample splits, full-cost comparisons and fixed-block device-noise simulation reported here are retrospective. Fitting, selection and evaluation are separated in the static and fixed-block studies.

\subsection{Accuracy and decision quality}

The CDR control in Fig.~\ref{fig:accuracy_decision_separation} isolates the effect of a positive affine correction $x\mapsto ax+b$, with slopes $a\in[1.292,1.478]$. This oracle control fits exact noisy means to ideal values on the same candidates. It reduces MSE by factors of $43$--$180$ over the reported evaluation budgets while preserving every ranking and tie. Its gap mean and covariance satisfy $\mu_{\CDR}=a\mu_{\Raw}$ and $\Sigma_{\CDR}=a^2\Sigma_{\Raw}$, so $\GRI_{\CDR}=\GRI_{\Raw}$.

The PEC effective-shot proxy in Fig.~\ref{fig:accuracy_decision_separation} isolates the overhead model of Section~\ref{sec:pec_signature}: it samples ideal outcomes with a reduced shot count and exhibits lower aggregate MSE but higher GRI than Raw. Appendix~\ref{app:signed_pec} gives its instance-level scope and the signed-sampling control below.

\begin{figure}[!ht]
    \centering
    \includegraphics[width=0.92\linewidth]{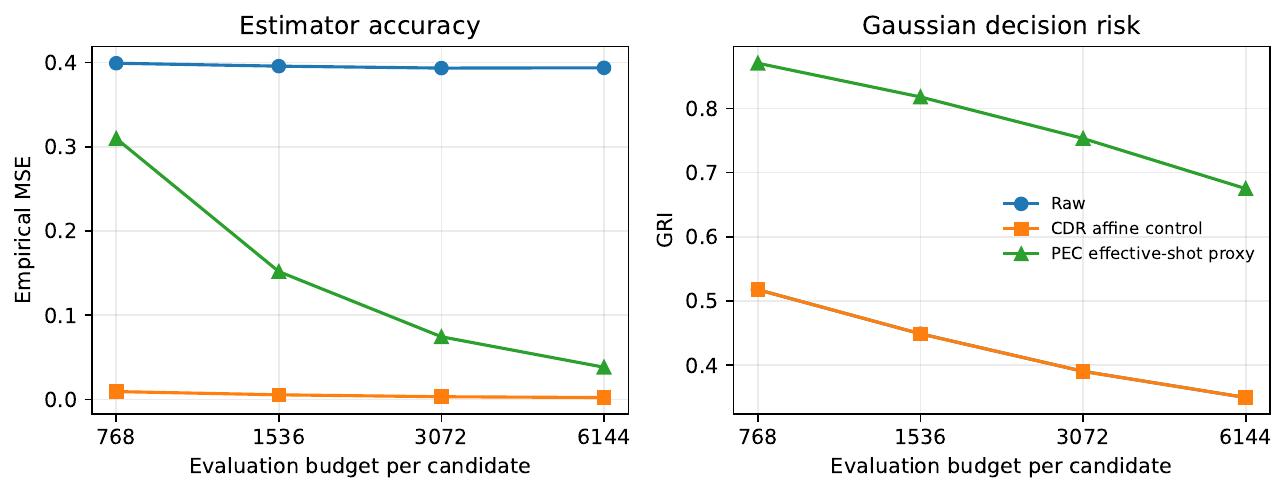}
    \caption{Accuracy--decision separation at the reported evaluation budgets. The oracle affine CDR control lowers MSE while preserving GRI; its GRI curve overlays Raw's. The PEC effective-shot proxy has lower MSE and higher GRI than Raw.}
    \label{fig:accuracy_decision_separation}
\end{figure}

\paragraph{Signed PEC at equal evaluation cost.}
A separate control samples the exact joint law of PEC signs and measurement outcomes on eight new six-qubit QAOA landscapes, with nine candidates, two known depolarizing noise levels and ideal Pauli corrections. The design and primary budget $B=3072$ shots per candidate were fixed before outcomes. Each method uses $9B$ shots per decision. At this budget, empirical MSE falls from $0.03975$ to $0.00831$, a $79\%$ reduction, while failure rises from $3.55\%$ to $14.28\%$ (Table~\ref{tab:signed_pec}). The paired difference is $10.73$ percentage points, with exploratory 90\% landscape-bootstrap interval $[3.82,18.26]$.

\begin{table}[!ht]
\centering
\caption{Signed PEC with a known noise channel. MSE is the finite-shot expectation, averaged equally over candidates, landscapes and noise levels; failure is observed in 8,192 decisions per method and budget. The primary budget is $B=3072$; the other rows are sensitivity comparisons. Channel learning is outside this controlled experiment.}
\label{tab:signed_pec}
\begin{tabular}{rrrrr}
\toprule
& \multicolumn{2}{c}{Expected MSE} & \multicolumn{2}{c}{Observed failure (\%)}\\
\cmidrule(lr){2-3}\cmidrule(lr){4-5}
$B$ & Raw & Signed PEC & Raw & Signed PEC\\
\midrule
768 & 0.043062 & 0.033354 & 12.15 & 28.53\\
3072 & 0.039780 & 0.008339 & 3.55 & 14.28\\
12288 & 0.038960 & 0.002085 & 0.10 & 5.53\\
\bottomrule
\end{tabular}
\end{table}

The signed estimator therefore combines lower pointwise MSE with more frequent selection errors in this control. Its effective margins and gap covariance describe the decision consequence of its bias and sampling variance. Gaussian risk computed from the known moments differs from the finite-shot discrete risk by at most $0.35$ percentage points across this control. Appendix~\ref{app:signed_pec} specifies the law, uncertainty calculation and simulation scope.

\FloatBarrier
\subsection{Covariance and risk prediction}

The next study tests whether the kernel construction reproduces Gaussian scores and whether those scores predict empirical failure ordering. With a common ideal mean, $R(\mu_{\rm ideal},\Sigma^{\rm pred}_m)$ and $R(\mu_{\rm ideal},\widehat\Sigma_m)$ have the same order on all 502 selected active pairs; 487 satisfy the original directional criterion and 15 do not. Against observed failures in the 40 archived repetitions, the same predictions give 302 strict matches, 133 reversals and 67 ties. A common-mean covariance comparison therefore captures only part of the method-dependent risk.

Estimating each method's effective mean and covariance from 20 repetitions, then evaluating failure on 20 disjoint repetitions, gives 310 strict matches, 94 reversals and 98 ties (Fig.~\ref{fig:pullback_directional}). The active-pair set was selected retrospectively, and pairs within an instance share observations; these counts describe dependent finite-sample comparisons. Further covariance and channel-learning diagnostics appear in Appendix~\ref{app:extended_experiments}.

\begin{figure}[!ht]
    \centering
    \includegraphics[width=0.82\linewidth]{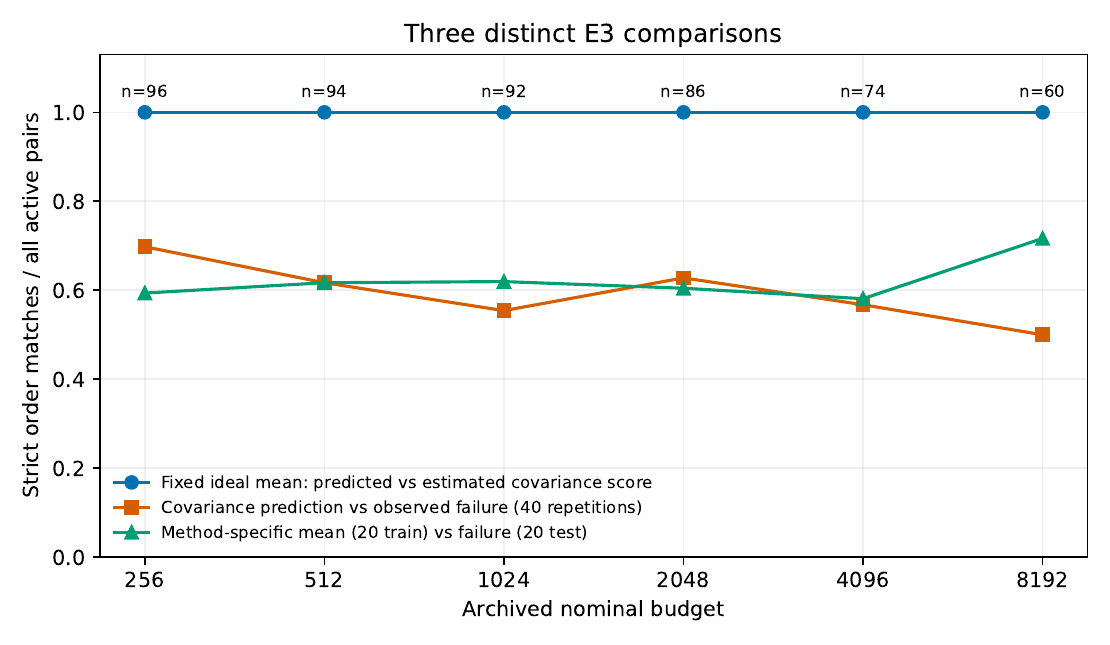}
    \caption{Covariance-score and empirical-failure comparisons on 502 selected active pairs. The method-specific mean and covariance are fitted on 20 repetitions and evaluated against failures on 20 disjoint repetitions. Ties remain in the denominators; the comparisons are descriptive and dependent within instances.}
    \label{fig:pullback_directional}
\end{figure}

\subsection{Method selection at matched shot costs}
\label{sec:e1_reanalysis}

The static study compares selection by estimated Gaussian risk with selection by MSE at matched shot costs, using all 48 archived C001--C048 instances. Within each instance, SHA256 sorting of 40 replicate IDs forms disjoint sets of 10 for coefficient fitting, 10 for method selection and 20 for evaluation, with the same assignment at both budgets. Fitted coefficients and method choices are fixed before test evaluation. The split protocol was fixed after the original aggregate outcomes were available.

This simulation uses ideal values to evaluate MSE and the ideal optimum as the reference for both selectors. Decision-ZNE coefficients use fitting-sample gap covariance and ideal gaps as tuning references. These simulation-oracle inputs specify the information available to the comparison. Risk estimation uses the 10 selection repetitions, Ledoit--Wolf covariance and a deterministic numerical Gaussian CDF.

\paragraph{Evaluation and learning budgets.}
Let $t\in\{512,1024\}$ denote the primitive prefix length and let $n=25$. In the evaluation-matched comparison, ZNE uses the first $t$ shots of each of three scales, while Raw uses scale-one shots $(t,4t]$. Each policy therefore has $3nt$ evaluation shots per test decision.
The budgets within one replicate can overlap, so $t=1024$ provides a sensitivity analysis of the primary $t=512$ comparison.

The full-cost comparison charges each selector $10(3nt)$ fitting shots, $10(6nt)$ selection shots and $20(3nt)$ evaluation shots, totaling $150nt$. Selection costs $6nt$ per repetition because the union of candidate-method records uses $4t$ scale-one shots and $t$ shots at each other scale; shared records are counted once. Learning costs are amortized over 20 independent decisions on the same fixed landscape. At $t=512$, every complete policy receives 1,920,000 shots per instance, or 96,000 per decision including learning. Always-Raw and fixed Accuracy-ZNE spend $7.5nt$ per evaluation decision. Fixed Decision-ZNE pays $30nt$ to fit and spends $6nt$ per evaluation decision.

Appendix~\ref{app:implementation_checks} gives the reconstruction of shot blocks from archived prefixes for both comparisons. The resource model counts shots and excludes classical computation and oracle evaluation of ideal quantities.

\begin{table}[!ht]
\centering\small
\caption{Static-selection test failure over 48 instances and 20 test repetitions each (960 decisions). ``Evaluation only'' equalizes target evaluation shots; ``Full cost'' also charges fitting and selection over the declared 20-decision cohort. Rates describe the fixed retrospective split; lower is better.}
\label{tab:e1_reanalysis}
\setlength{\tabcolsep}{4pt}
\begin{tabular}{rlrrrrr}\toprule
$t$ & Accounting & Raw & Acc.-ZNE & Dec.-ZNE & MSE sel. & GRI sel.\\\midrule
512 & Evaluation only & 0.3125 & 0.4729 & 0.4792 & 0.4719 & 0.3083 \\
512 & Full cost & 0.2333 & 0.3510 & 0.3865 & 0.4719 & 0.3083 \\
1024 & Evaluation only & 0.2521 & 0.3906 & 0.4125 & 0.3927 & 0.2615 \\
1024 & Full cost & 0.1865 & 0.2823 & 0.3156 & 0.3927 & 0.2615 \\
\bottomrule\end{tabular}
\end{table}

\paragraph{Results.}
At $t=512$, evaluation-matched failure is $0.3083$ for GRI selection and $0.3125$ for always-Raw. Their paired difference, GRI minus Raw, is $-0.0042$, with descriptive 95\% instance-cluster bootstrap interval $[-0.0198,0.0115]$. At equal full cost, Raw reaches $0.2333$ and GRI selection remains at $0.3083$, giving difference $+0.0750$ with interval $[0.0417,0.1094]$. GRI selection improves over MSE selection ($0.4719$), with difference $-0.1635$ and interval $[-0.2188,-0.1115]$. At $t=1024$, the evaluation-matched GRI--Raw interval again includes zero, while the full-cost difference is $+0.0750$ with interval $[0.0344,0.1188]$.

Intervals use 10,000 resamples of the 48 complete instance clusters, retaining the pairing of policies and test repetitions. They describe the fixed split and fitted policies. The MSE selector chooses a ZNE variant in every instance at both budgets, whereas GRI predominantly chooses Raw. Figure~\ref{fig:e1_reanalysis} shows how both method choice and resource allocation enter the comparison: the selectors use part of their budget to learn the choice, while always-Raw allocates it to evaluation. At equal full shot cost, Raw has the lowest observed failure rate among these policies.

Decision-ZNE also has higher observed failure than Accuracy-ZNE at both budgets and under both accounting conventions. The fitted coefficient rule uses ideal gaps as tuning references without the certified lower bounds on mitigated effective margins required by the design in Section~\ref{sec:qem_design}. Its observed performance describes this implementation of the covariance objective. The selection results therefore concern both the available method menu and the information used to fit and compare its members.

\paragraph{Sensitivity to the data split.}
Repeating fitting, selection and testing over 20 balanced $10/10/20$ partitions of the same 48 instances preserves both full-cost directions in every partition at both budgets. At $t=512$, GRI-minus-Raw failure differences range from $6.35$ to $10.21$ percentage points, and GRI-minus-MSE-selection differences from $-18.13$ to $-12.60$. These are ranges across dependent partitions, not confidence intervals. Appendix~\ref{app:e1_split_sensitivity} describes the construction and cluster analysis.

\begin{figure}[!ht]
\centering\includegraphics[width=.96\linewidth]{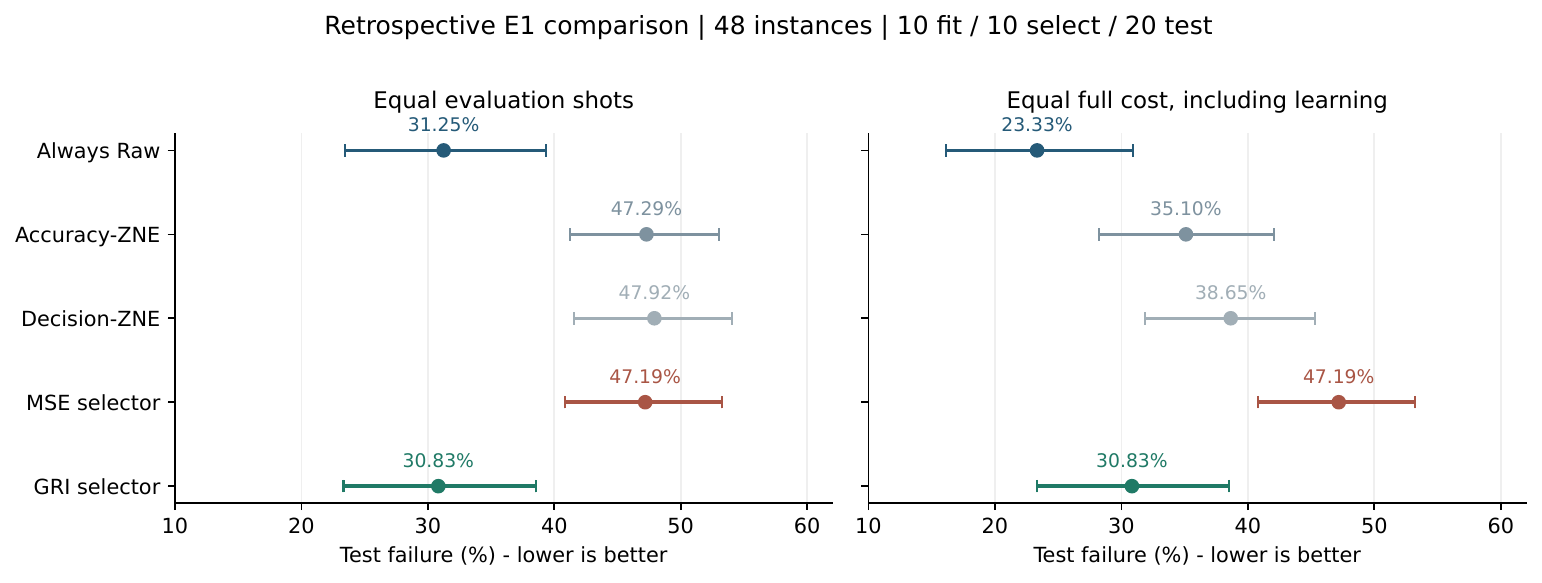}
\caption{Static selection at equal shot costs. Points show failure rates on the independent test split; intervals are descriptive 95\% instance-cluster bootstrap intervals. The left panel equalizes evaluation shots. The right panel also charges fitting and selection, allowing untuned baselines to devote those resources to evaluation.}
\label{fig:e1_reanalysis}
\end{figure}

\FloatBarrier
\subsection{Dynamic selection and total shot cost}

The dynamic experiment applies the archived static method choices to repeated decisions along a shared trajectory. These choices retain the original training procedure; they are separate from the retrospective static split in Section~\ref{sec:e1_reanalysis}. The next proposal is chosen using Raw at $B_{\max}=8192$, including at smaller reported budgets. Over C001--C032, aggregate success reaches at most $0.7758$ for decision-aware selection (DA) and $0.7883$ for MSE selection, both below the target $0.80$.

The total shot cost includes Raw at $B_{\max}$ for the 19 proposal updates that affect later rounds, the selected method's additional scales and the final-round evaluation. It also charges static-selector training (40 repetitions, three scales and 25 candidates at the source budget), amortized over 40 trajectories. Optimizing the target crossing by full shot cost gives a median MSE/DA cost ratio of $1.30$ on the 13 instances where both policies cross the target. Analytic coefficient construction and ideal truth are oracle inputs excluded from the shot cost.

\begin{figure}[!ht]
    \centering
    \includegraphics[width=0.82\linewidth]{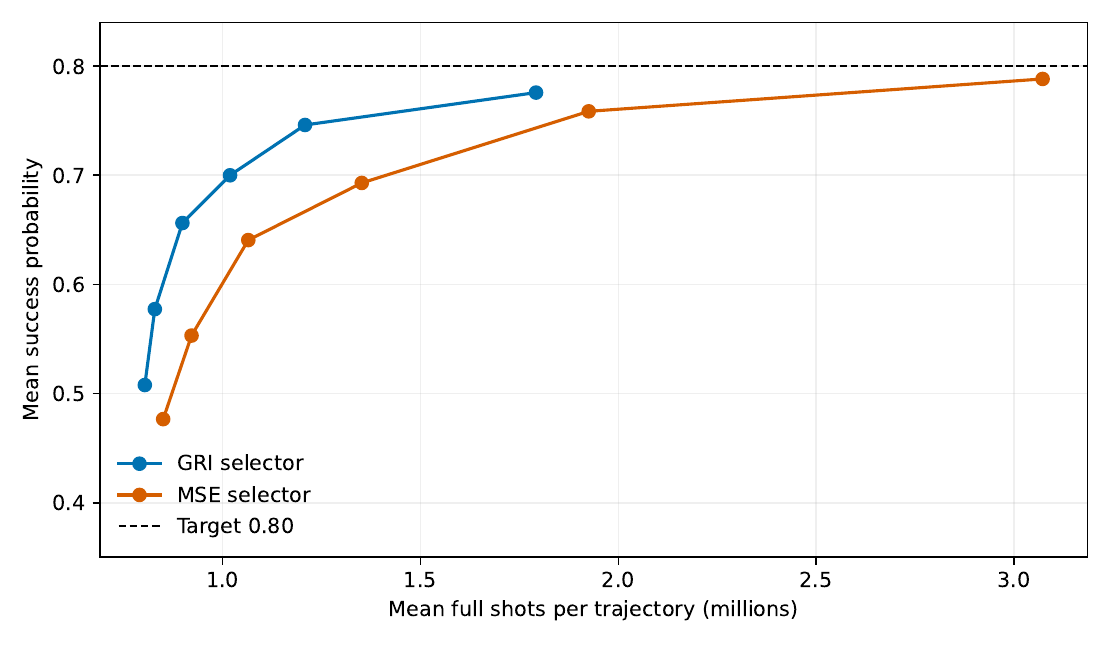}
    \caption{Dynamic evaluation on the shared Raw-at-8192 proposal trajectory. The horizontal axis charges mean full shots per trajectory, including amortized static-selector training. Points average 32 instances at each nominal budget, with costs varying across instances. The reference success probability is $0.80$; the observed maxima are $0.7758$ for decision-aware selection and $0.7883$ for MSE selection.}
    \label{fig:dynamic_success}
\end{figure}

\subsection{Noise-model controls and device tests}

\paragraph{Noise-model controls.}
To examine sensitivity to the noise model, a heterogeneous-PEC Gaussian proxy uses the ideal mean and covariance $\gamma^2\Sigma_{\Raw}$, with overhead $\gamma^2=88.94$. Its GRI exceeds Raw's throughout the evaluated budgets. This proxy scales noisy covariance, whereas Fig.~\ref{fig:accuracy_decision_separation} uses ideal multinomial samples. Under amplitude damping plus phase damping, a separate calculation uses known moments and $B$ shots per candidate and scale: Raw costs $B$ and ZNE costs $3B$. The analytic decision-aware prediction beats Raw on $0/16$ instances at every budget, so the planned empirical stage was not triggered. These controls identify regimes in which the available correction and its sampling cost provide no decision gain.

\paragraph{Calibrated device-noise simulation.}
A four-method \texttt{FakeKolkataV2} simulation uses Raw, Raw+REM, Accuracy-ZNE and Decision-ZNE on all 48 original landscapes. Routing is restricted to the 10-qubit physical block $[16,22,25,19,20,24,23,14,26,13]$, previously specified for P004, to make density-matrix simulation tractable. Sampling separates 10 coefficient-fitting, 10 method-selection and 20 test repetitions. At nominal evaluation budget $512$, ZNE uses $3\lfloor512/3\rfloor=510$ shots per candidate, while Raw and REM use 512. Candidate labels and ideal gaps remain simulation-oracle inputs.

On the 960 test decisions, failure is $0.6708$ for MSE selection and $0.5031$ for GRI selection. The relative reduction is $0.2500$, with descriptive 95\% paired instance-cluster bootstrap interval $[0.1716,0.3333]$ from 10,000 resamples. Always-Raw has failure $0.5042$; the paired GRI-minus-Raw difference is $-0.0010$, with interval $[-0.0240,0.0198]$. These intervals condition on the fitted coefficients and selection records. Total policy cost includes fitting, acquisition of the full selection menu, with the ZNE records shared between its two arms, and evaluation of the selected method. Thus the displayed evaluation-budget comparison leaves different end-to-end policy costs.

\paragraph{Hardware covariance propagation.}
The covariance calculation is also tested in a preregistered exploratory four-qubit micro-cell on a 156-qubit device, comparing predicted and measured decision covariances. Relative Frobenius errors range from 13\% to 21\%. Projected ranking-error bounds exceed the prespecified threshold at this repetition budget, leaving decision ordering unresolved. Appendix~\ref{app:hardware_microcell} gives the protocol, qubit-selection procedure and results.

\section{Discussion}

The controls describe how an estimation correction acts on candidate differences. A common positive affine CDR map rescales all gaps and preserves their ordering. Signed PEC changes both their bias and uncertainty, producing lower MSE and more selection errors at the primary budget. These are consequences of the estimator and measurement model: MSE summarizes the squared errors in the values, while selection risk depends on their joint effect near a decision boundary.

Method selection uses measurements both to compare estimators and to evaluate the selected one. The static study includes fitting and selection in a budget shared across 20 decisions on each landscape. Its results show how the choice criterion and the allocation of learning shots jointly affect the final selection error. The comparisons with fixed methods place these two uses of measurements on the same shot budget.

The information available to a procedure also enters this description. The static selectors use ideal simulation references; the signed-PEC control evaluates a stipulated stationary channel with ideal corrections. The readout example in Section~\ref{sec:estimating_decision_kernel} shows how target records can leave the ideal minimizer ambiguous when the channel is unknown. Calibration and reference acquisition connect the inputs required by a decision rule to the physical records and sampling resources used to obtain them.

\section{Conclusion}
The joint residual gap law connects QEM estimation to decisions among candidates. Its Gaussian representation expresses the effects of bias, shared noise and shot allocation through effective margins and covariance, and supports extrapolation design when positive margins are certified.

The experiments relate this description to method choice and sampling cost. CDR and PEC controls illustrate distinct relationships between MSE and selection error. In the retrospective static comparison, Gaussian-risk selection has lower failure than MSE selection, while Raw attains the lowest failure at equal full shot cost. These results connect mitigation performance to the estimation objective, the downstream decision and the shots used for learning and evaluation.

\Needspace{10\baselineskip}
\section*{Author contributions}

Vicenzo Scavino Alfaro conceived the project, developed the theory, implemented and validated the numerical and hardware workflows, analyzed the data, and wrote the manuscript.  Large language model tools were used as interactive assistance for code editing, documentation, literature checks, and language polishing; the author reviewed and validated all text, code, calculations, citations, and conclusions and is fully responsible for the work.

\section*{Acknowledgments}

I thank my parents for their unconditional personal and material support; their help is what makes it possible to carry out this work as an independent researcher.

This work was carried out entirely with free and openly available tools.  I acknowledge the open-source scientific software community, in particular Qiskit and Qiskit Aer, NumPy, SciPy, pandas, NetworkX, and the LaTeX ecosystem.  I thank IBM Quantum for promotional access to quantum computing services, used for the hardware micro-cell experiment.  The views expressed are those of the author and do not reflect the official policy or position of IBM or the IBM Quantum team.  No external funding was received for this work.

\section*{Code and data availability}

The code, archived inputs and analysis outputs supporting this article are available as release \texttt{v1.2.0} of \url{https://github.com/vicenzoscavino1999/decision-aware-qem}, archived at \href{https://doi.org/10.5281/zenodo.23608223}{10.5281/zenodo.23608223}. The reproduction workflow regenerates the reported summaries and figures from stored records. The original E2 measurement rounds are available as a separate dataset at \href{https://doi.org/10.5281/zenodo.23607497}{10.5281/zenodo.23607497}. The code package includes the reproduction guides, pinned environment and file manifests.

Separate supplements provide the signed-PEC control and E1 partition sensitivity, including their designs, input records, analysis scripts and results. The manuscript source and provenance metadata accompany the analysis package; Appendix~\ref{app:implementation_checks} describes the numerical methods used for the reported results.

\appendix

\section{Estimation and stability}
\label{app:estimation_stability}

On a fixed critical band of dimension $d$, obtain $R\ge2$ independent repetitions at budget $B$ and form centered per-unit gaps $G_r=\sqrt B(\widehat\Delta_{m,r}-\mu_m)$. A reference error must be $o(B^{-1/2})$ or explicitly accounted for. Empirical centering estimates covariance about an unknown mean; it does not identify shared drift or remove uncertainty in a common pilot reference. The sample covariance is
\[
\widehat\Sigma=\frac1{R-1}\sum_{r=1}^R(G_r-\overline G)(G_r-\overline G)^\top.
\]
The following proposition writes $\Sigma$ for the common covariance of these rescaled finite-budget gaps, namely $\Sigma_{m,B}$. Approximating the limiting kernel $\Sigma_m$ adds the deterministic error $\|\Sigma_{m,B}-\Sigma_m\|_{\op}$ to the covariance-estimation bound.

\begin{proposition}[Covariance and score stability]
\label{thm:sample_covariance}\label{thm:ranking_stability}\label{thm:plugin_risk_stability}
For independent centered $G_r$ with common covariance $\Sigma$ and directional sub-Gaussian norm $\sup_{\|u\|_2=1}\|u^\top G_r\|_{\psi_2}\le\sigma$, with probability at least $1-\delta$,
\begin{equation}\label{eq:sample_cov_bound}
\|\widehat\Sigma-\Sigma\|_{\op}
\le C\sigma^2\left[\sqrt{\frac{d+\log(2/\delta)}{R}}
+\frac{d+\log(2/\delta)}{R}\right].
\end{equation}
Here $C$ is universal and $\sigma$ controls all directions, not just coordinate variances. For estimated margins and covariance satisfying $\|\widehat\mu-\mu\|_\infty\le\eta_\mu$ and $\|\widehat\Sigma-\Sigma\|_{\op}\le\eta_\Sigma$, two deterministic consequences hold:
\begin{enumerate}[label=(\alph*)]
\item If $\Sigma_{ii}\ge s_{\min}^2>0$ and $\eta_\Sigma<s_{\min}^2/2$, then
\begin{equation}\label{eq:dri_stability}
|\widehat\DRI^{\rm unit}-\DRI^{\rm unit}|
\le \frac{\sqrt2\,\eta_\mu}{s_{\min}}
+\frac{\sqrt2\,\|\mu\|_\infty\eta_\Sigma}{s_{\min}^3}.
\end{equation}
A two-method DRI ordering is preserved when its true separation exceeds the sum of the two error bounds, on their joint estimation event.
\item If $\Sigma\succ0$ has spectrum in $[\lambda_{\min},\lambda_{\max}]$ and $\eta_\Sigma<\lambda_{\min}/2$, the Gaussian risk in Eq.~\eqref{eq:gaussian_risk} obeys
\begin{equation}\label{eq:gri_stability}
|R_B(\widehat\mu,\widehat\Sigma)-R_B(\mu,\Sigma)|
\le C_{d,B}(\eta_\mu+\eta_\Sigma).
\end{equation}
The constant depends on the dimension, budget and covariance bounds (and may be taken to depend on the margin neighborhood); it is not asserted uniform as $B\to\infty$.
\end{enumerate}
\end{proposition}

\begin{proof}
Equation~\eqref{eq:sample_cov_bound} is sample-covariance concentration for sub-Gaussian vectors \cite{Vershynin2018,KoltchinskiiLounici2017}. For (a), $\widehat v_i\ge s_{\min}^2/2$; applying the mean-value bound for $v^{-1/2}$ to each standardized margin gives Eq.~\eqref{eq:dri_stability}.

Taking minima cannot increase the largest coordinate error; intersect the estimation events or apply a union bound for simultaneous rankings. For (b), integrate the Gaussian density derivatives along the segment joining the two parameters. Uniform positive definiteness bounds their integrals and gives local Lipschitz continuity. Gaussian-max comparison yields sharper refinements \cite{Chernozhukov2013}.
\end{proof}

\paragraph{Readout calibration and identifiability.}
Readout mitigation corrects Pauli expectations using response factors learned from calibration \cite{Bravyi2021,vandenBerg2022}. Consider $H=Z_0Z_1+Z_1Z_2$ and independent, symmetric readout flips with probabilities $p_j<1/2$, constant across candidates and shots. Put $a_j=1-2p_j>0$. For observed parity means $u_x=\E[\widetilde Z_0\widetilde Z_1]$ and $v_x=\E[\widetilde Z_1\widetilde Z_2]$, the ideal energy is
\begin{equation}\label{eq:readout_relative_weights}
F(x)=\frac{u_x}{a_0a_1}+\frac{v_x}{a_1a_2}
=\frac{u_x+\eta v_x}{a_0a_1},\qquad \eta=\frac{a_0}{a_2}.
\end{equation}
Thus the population ordering depends on the relative weight $\eta$; the common positive scale cancels. The following example shows why that weight need not be identifiable from the target records alone.

\begin{example}[Identical target records, opposite ideal minimizers]
\label{ex:readout_identifiability}
Take two candidates with observed means $(u_0,v_0)=(0.30,0.05)$ and $(u_1,v_1)=(0.10,0.30)$. Two admissible state--channel models give
\begin{center}
\begin{tabular}{ccccc}
\toprule
Model & $(a_0,a_1,a_2)$ & $F(x_0)$ & $F(x_1)$ & Ideal minimizer\\
\midrule
$A$ & $(0.8,1,0.5)$ & $0.475$ & $0.725$ & $x_0$\\
$B$ & $(0.4,1,0.8)$ & $0.8125$ & $0.625$ & $x_1$\\
\bottomrule
\end{tabular}
\end{center}
For either channel, choose a diagonal target state supported on $z_1=+1$, with probabilities
\[
P_x^{(a)}(z_0,+1,z_2)=\tfrac14\bigl(1+(u_x/a_0)z_0\bigr)
\bigl(1+(v_x/a_2)z_2\bigr),\qquad z_0,z_2\in\{-1,+1\}.
\]
These probabilities are nonnegative and sum to one. Applying the independent readout flips gives, in both models,
\[
Q_x(z_0,+1,z_2)=\tfrac14(1+u_xz_0)(1+v_xz_2),
\]
with zero mass at $z_1=-1$. Hence the complete bitstring laws agree for each candidate. With independent shots and any fixed allocation, every rule $\delta$ selecting one of the two candidates from the target records, including a randomized rule, satisfies
\[
\Prob_A(\delta\ne x_0)+\Prob_B(\delta\ne x_1)=1.
\]
At least one model therefore has error probability at least $1/2$, for every target-shot budget.
\end{example}

The unknown ideal states vary between these models; the observation model supplies target records without additional information fixing those states. A correctly prepared reference state $|000\rangle$, measured under the same channel, has parity means $a_0a_1$ and $a_1a_2$, whose ratio is $\eta$. Such a reference distinguishes the two models. At finite calibration budget, the ratio's uncertainty depends on the underlying attenuations and must be propagated into the decision risk, with calibration shots charged to the full cost.

\section{Covariance linearization and normal approximation}
\label{app:covariance_clt}

\subsection{Nonlinear maps}

For $\widehat F_m=g_m(Y(B))$, write $\bar Y_B=\E Y(B)$ and $J_m=Dg_m(\bar Y_B)$. Assume $B\Cov(Y(B))=O(1)$, uniformly bounded $J_m$, and
\[
g_m(Y(B))=g_m(\bar Y_B)+J_m(Y(B)-\bar Y_B)+r_B,
\qquad \E\|r_B\|_2^2=o(B^{-1}).
\]
Then $\Cov(r_B)=o(B^{-1})$ and Cauchy--Schwarz bounds its cross-covariance with the linear term by $o(B^{-1})$. Consequently
\begin{equation}
K_m^{(B)}=J_mC_{\rm dev}(B)J_m^\top+o(B^{-1}).
\end{equation}
The mean-square remainder requires control beyond a distributional delta expansion, including outside the local neighborhood. For learned maps, condition on training or include training and target primitives in $Y(B)$.

\subsection{Shared calibration and tail control}
\label{app:calibration_covariance}
A uniformly bounded Hessian on all segments from $\bar p$ to $p_B$, together with $\E\|\eta_B\|_2^4=o(B^{-1})$, suffices for the mean-square remainder in Section~\ref{subsec:anchoring_cdev}. Covariance scaling alone is insufficient.

Cauchy--Schwarz and $\E\|r_B\|_2^2=o(B^{-1})$ give
\[
\Cov(Y(B))=J_p\Cov(\eta_B)J_p^\top+\Cov(\varepsilon_B)
+J_p\Cov(\eta_B,\varepsilon_B)+\Cov(\varepsilon_B,\eta_B)J_p^\top
+o(B^{-1}).
\]
The remainder has negligible cross-covariance because the other terms have covariance $O(B^{-1})$. Conditional centering makes both cross terms zero exactly and gives $\Cov(\varepsilon_B)=\E\Cov(\varepsilon_B\mid p_B)$ by total covariance. This proves Eq.~\eqref{eq:cdev_physical_delta} and its shared-data extension without a Gaussian assumption. Taylor's bound $\|r_B\|_2\le C\|\eta_B\|_2^2$ proves the fourth-moment sufficient condition. For example, $\eta_B=B^{-1/2}U$ with $\E\|U\|_2^4<\infty$ meets it when the Hessian bound holds.

When reused calibration and target data invalidate conditional centering, assume instead $B\Cov(\varepsilon_B)\to D_{\rm shot}$ and $B\Cov(\eta_B,\varepsilon_B)\to C_{p\varepsilon}$. Add $J_pC_{p\varepsilon}+C_{p\varepsilon}^\top J_p^\top$ to Eq.~\eqref{eq:cdev_physical_delta}.

To see why second moments alone are insufficient, take $\bar p=1/2$, with $p_B=1/2$ with probability $1-1/B$ and $p_B=0,1$ each with probability $1/(2B)$. Then $\E\eta_B=0$ and $\Var(\eta_B)=1/(4B)$. Let $f(p)=(2p-1)^2$ and conditionally sample $B$ independent $\{-1,+1\}$ outcomes of mean $f(p_B)$. Here $J_p=0$ and $D_{\rm shot}=1$, yet
\[
B\Var(Y(B))=2-2/B\longrightarrow2,
\qquad \E r_B^2=1/B.
\]
The rare calibration excursions contribute at leading order and violate precisely the required mean-square remainder condition.

\paragraph{Persistent calibration variability.}
Persistent calibration variability, for example $\Cov(\eta_B)\to V_\infty\ne0$, lies outside the shrinking-noise assumptions above and requires the full response distribution. If $p_B$ converges in distribution to a random $p_\infty$ and $f$ is bounded and continuous, its covariance contribution converges to $\Cov(f(p_\infty))$. A surviving gap component can produce a risk floor, but parameter covariance alone does not imply one. Under global depolarizing response $f_{x,s}(p)=\alpha_s(p)F(x)$, the small-noise calibration term describes correlated fluctuations proportional to landscape values. Their effect depends on the differential projection $LM_mJ_p$, not simply on their ambient size.

\subsection{Degenerate kernels and finite-budget risk}
\label{app:finite_budget_degeneracy}
For the fixed linear estimator $\widehat F_m=M_mY(B)$ and a fixed contrast $v$, the finite-budget identity is
\[
\Var(v^\top L\widehat F_m^{(B)})
=\|C_{\rm dev}(B)^{1/2}M_m^\top L^\top v\|_2^2.
\]
This is zero if and only if the displayed vector vanishes.

For independent Rademacher observations $X_1,\ldots,X_B$, form overlapping means
\[
Y_1=\frac1{B-1}\sum_{j=1}^{B-1}X_j,
\qquad Y_2=\frac1{B-1}\sum_{j=2}^{B}X_j.
\]
Their variances are $1/(B-1)$ and their covariance is $(B-2)/(B-1)^2$, so $B\Cov(Y)\to\1\1^\top$. The contrast $D_B=Y_2-Y_1=(X_B-X_1)/(B-1)$ has zero limiting kernel but $\Var(D_B)=2/(B-1)^2>0$ at every finite budget. It takes values $-2/(B-1),0,2/(B-1)$ with probabilities $1/4,1/2,1/4$. Hence $\Prob[\mu+D_B\le0]=1/4$ for the fixed margin $\mu=0.1$ and $3\le B\le21$; for $\mu_B=1/(B-1)$ this probability is $1/4$ at every $B\ge3$. At zero margin it is $3/4$, whereas an exactly deterministic zero gap fails with probability one. Noise of smaller order can therefore remain decision-relevant near a boundary.

There is a distribution-free way to bound an omitted positive-mean gap. Write $\widehat\Delta_i^{(B)}=\mu_i^{(B)}+X_{i,B}$ with $\E X_{i,B}=0$ and $v_{i,B}=\Var(X_{i,B})$. For $\mu_i^{(B)}>0$, Cantelli's inequality gives
\[
\Prob[\widehat\Delta_i^{(B)}\le0]\le
\frac{v_{i,B}}{v_{i,B}+(\mu_i^{(B)})^2}.
\]
Indeed, for $t\ge0$, Markov's inequality applied to $(X_{i,B}-t)^2$ gives the upper bound $(v_{i,B}+t^2)/(\mu_i^{(B)}+t)^2$; taking $t=v_{i,B}/\mu_i^{(B)}$ proves the claim, including $v_{i,B}=0$. For a fixed partition into retained coordinates $\calC$ and omitted coordinates $\mathcal{O}$, with positive means for all omitted coordinates, a union bound yields
\begin{equation}\label{eq:omitted_gap_bound}
R_{\calC}(B)\le R(B)\le
\min\left\{1,\ R_{\calC}(B)+\sum_{i\in\mathcal{O}}
\frac{v_{i,B}}{v_{i,B}+(\mu_i^{(B)})^2}\right\}.
\end{equation}
Here $R_{\calC}$ is the actual restricted risk, not automatically its Gaussian estimate. No independence or normality is required. Certified upper variance bounds and positive lower mean bounds can replace $v_{i,B}$ and $\mu_i^{(B)}$ on their joint validity event; unqualified sample estimates do not certify this inequality. A known zero finite variance makes a positive-mean gap harmless and a nonpositive-mean gap an almost-sure failure. If instead $v_{i,B}=o(B^{-1})$ and the mean stays bounded away from zero, the bound is $o(B^{-1})$, without establishing an exponential rate. For shrinking margins, the sufficient condition is $v_{i,B}/(\mu_i^{(B)})^2\to0$.

\Needspace{12\baselineskip}
\subsection{Finite-shot normal approximation}
\label{app:normal_approximation}

\begin{proposition}[Finite-shot gap CLT]
\label{prop:gap_clt}
Suppose a fixed finite number of primitive groups contain jointly independent, uniformly bounded outcomes, with non-adaptive allocations $B_j/B\to\lambda_j>0$. Let $\overline Y$ collect their means, $\|M_m(B)\|_{\op}=O(1)$, and
\[
\Sigma_{m,B}=(LM_m(B))B\Cov(\overline Y)(LM_m(B))^\top\longrightarrow\Sigma_m.
\]
For $L\xi_m=(LM_m(B))(\overline Y-\E\overline Y)$,
\begin{equation}
\sqrt B\,L\xi_m\ \overset{d}{\longrightarrow}\ \calN(0,\Sigma_m).
\end{equation}
If $\Sigma_{m,B}$ is uniformly nondegenerate on a fixed effective support, the convex-set approximation error relative to $\calN(0,\Sigma_{m,B})$ is $O(B^{-1/2})$ in fixed dimension. The same order applies to a fixed finite union of polyhedra. Replacement by $\Sigma_m$ adds a covariance-comparison error: convergence alone gives $o(1)$, whereas an $O(B^{-1/2})$ rate follows, for example, from $\|\Sigma_{m,B}-\Sigma_m\|=O(B^{-1/2})$ under uniform conditioning.
\end{proposition}

\begin{proof}
Positive allocation fractions and bounded maps make the centered normalized per-shot summands uniformly $O(B^{-1/2})$, so the triangular-array CLT applies with the stated covariance limit.

For the quantitative bound, use orthonormal coordinates on the fixed effective support of dimension $r\ge1$. Let $X_{B,k}$ be the per-shot contributions to $\sqrt B\,L\xi_m$ in these coordinates, and let $V_B$ be the restriction of $\Sigma_{m,B}$ to this support. For $W_{B,k}=V_B^{-1/2}X_{B,k}$, uniform conditioning gives $\max_k\|W_{B,k}\|\le C B^{-1/2}$ almost surely, with $C$ independent of $B$. Independence and centering give $\sum_k\Cov(W_{B,k})=\mathrm{Id}_r$, hence
\[
\beta_B:=\sum_k\E\|W_{B,k}\|^3
\le\frac{C}{\sqrt B}\sum_k\E\|W_{B,k}\|^2
=\frac{Cr}{\sqrt B}.
\]
Bentkus's bound for independent, not necessarily identically distributed vectors is at most $C_0r^{1/4}\beta_B$ over convex sets \cite{Bentkus2005}, giving the stated rate for fixed $r$. If the support has dimension zero, the centered vector and its Gaussian approximation both vanish almost surely. Inclusion--exclusion extends the bound to a fixed finite union of polyhedra. Comparing the two Gaussian covariances adds the final error term.
\end{proof}

This is an absolute-probability approximation. It does not supply the exponentially small tail probabilities used by Gaussian large-deviation rates.

\section{Gaussian rates, comparisons and allocation}
\label{app:gaussian_proofs}

\subsection{General failure regions and active faces}
For a Borel success set $\calS$, put $\calF_\mu=\{y:\mu+y\notin\calS\}$ and $\calR_\Sigma=\operatorname{range}(\Sigma)$. Define $J(y)=\tfrac12 y^\top\Sigma^\dagger y$ on $\calR_\Sigma$ and $J(y)=+\infty$ outside it. The candidate rate is
\begin{equation}\label{eq:decision_exponent}
I_{\rm cand}(\mu,\Sigma;\calS)=\inf_{y\in\calF_\mu\cap\calR_\Sigma}J(y),
\end{equation}
with an empty-set infimum equal to $+\infty$.

\begin{proposition}[Regular Gaussian failure regions]
\label{prop:general_gaussian_rate}\label{lem:jregularity}
Let $Z\sim\calN(0,\Sigma)$ and $A=\calF_\mu\cap\calR_\Sigma$, and put
\[
a_B=-B^{-1}\log\Prob[B^{-1/2}Z\in\calF_\mu].
\]
Relative closure and interior on $\calR_\Sigma$ give
\begin{equation}
\inf_{\cl A}J\le\liminf_B a_B\le\limsup_B a_B\le\inf_{\interior A}J.
\end{equation}
If the outer infima agree ($J$-regularity), the rate exists and equals Eq.~\eqref{eq:decision_exponent}. A finite union of halfspaces $h_k^\top y\le-c_k$ is regular when $c_k>0$ and $h_k^\top\Sigma h_k>0$ for each component.
\end{proposition}

\begin{proof}
The bounds are the Gaussian large-deviation principle on the support subspace \cite{DemboZeitouni1998}. For one halfspace, its minimizer is $y^\star=-c\Sigma h/(h^\top\Sigma h)$ with cost $c^2/(2h^\top\Sigma h)$. Interior points approach this boundary minimizer, so the infima agree. A finite union takes the minimum component rate. This includes the noisy-coordinate argmin formula in Eq.~\eqref{eq:argmin_exponent}.
\end{proof}

\begin{theorem}[Active faces and gap correlations]
\label{thm:diagonal_offdiagonal}
Let $Z\sim\calN(0,\Sigma)$ with $\Sigma\succ0$. Consider a full-dimensional polyhedral failure component whose active inequalities at the rate minimizer are $H_SZ\le-\sqrt B\,c_S$. Assume $H_S$ has full row rank, $c_S>0$, all remaining inequalities are strictly inactive there, and $\alpha_S=(H_S\Sigma H_S^\top)^{-1}c_S>0$ coordinatewise. Its rate is
\begin{equation}\label{eq:general_active_face_exponent}
I_S=\tfrac12c_S^\top(H_S\Sigma H_S^\top)^{-1}c_S.
\end{equation}
For a finite union of such regular components, take the smallest component rate. For coordinate constraints, this becomes $I_S=\tfrac12\mu_S^\top\Sigma_{SS}^{-1}\mu_S$.
\end{theorem}

\begin{proof}
The Gaussian rate minimizes $\tfrac12y^\top\Sigma^{-1}y$ on the failure component. The KKT equations give $y^\star=-\Sigma H_S^\top\alpha_S$; substitution gives Eq.~\eqref{eq:general_active_face_exponent}. Strict inactivity of the other constraints and full row rank allow interior points to approach $y^\star$, so the large-deviation lower and upper rates agree. A finite union takes the minimum.
\end{proof}

Negative multipliers invalidate a proposed active set; zero multipliers can leave its quadratic rate valid while changing the prefactor.

\paragraph{Finite-budget prefactors.}
For a coordinate active set $S$ with positive margins and $\Sigma_{SS}\succ0$, write $k=|S|$, $I_S=\tfrac12\mu_S^\top\Sigma_{SS}^{-1}\mu_S$ and $\alpha_S=\Sigma_{SS}^{-1}\mu_S>0$ componentwise. For some $0<A_S<C_S<\infty$ and sufficiently large $B$,
\begin{equation}\label{eq:regular_face_prefactor}
A_SB^{-k/2}e^{-BI_S}\le\Prob[Z_S\le-\sqrt B\mu_S]
\le C_SB^{-k/2}e^{-BI_S}.
\end{equation}
To see this, substitute $z=-\sqrt B\mu_S-u/\sqrt B$ in the Gaussian density. After extracting a positive constant times $B^{-k/2}e^{-BI_S}$, the remaining integral is $\int_{\R_+^k}\exp[-\alpha_S^\top u-u^\top\Sigma_{SS}^{-1}u/(2B)]\,du$. It converges by domination to the finite positive value $\prod_i\alpha_{S,i}^{-1}$, proving the bounds.

Strict positivity matters. For $\Sigma=\bigl(\begin{smallmatrix}1&1/2\\1/2&1\end{smallmatrix}\bigr)$ and $\mu=(1,1/2)^\top$, the minimizer is $-\mu$ and both constraints bind, but $\Sigma^{-1}\mu=(1,0)^\top$. Conditioning on the first coordinate gives $\Prob(Z_1\le-\sqrt B,Z_2\le-\sqrt B/2)\sim\tfrac12\Phi(-\sqrt B)$: the prefactor is $B^{-1/2}$ rather than $B^{-1}$.

\subsection{Local risk changes and Gaussian comparison}
\label{app:gaussian_comparison}\label{app:slepian_safe}

\paragraph{Local covariance sensitivity.}
On the positive-definite cone, Gaussian density differentiation or Plackett's formula \cite{Plackett1954} gives the symmetric risk gradient $G_m=\nabla_\Sigma R(\mu_m,\Sigma_m)$. Taylor expansion for symmetric $E_m\to0$, with $\Sigma_m+E_m\succ0$, yields
\begin{equation}\label{eq:local_risk_change}
\Delta_{AB}^{\rm pert}-\Delta_{AB}
=\langle G_A,E_A\rangle_F-\langle G_B,E_B\rangle_F
+o(\|(E_A,E_B)\|),
\end{equation}
where $\Delta_{AB}=R(\mu_A,\Sigma_A)-R(\mu_B,\Sigma_B)\ne0$. If the displayed linear correction has magnitude at most $\rho|\Delta_{AB}|$, $0<\rho<1$, it preserves the ranking to first order. A finite-perturbation guarantee also needs a remainder bound. The hardware diagnostic uses $\rho=0.25$; it is a local sensitivity check, including for Loewner-incomparable perturbations.

\paragraph{Comparison of experiments.}
For $\mathsf E_m=\{\calN(\theta,\Sigma_m):\theta\in\Theta\}$, if $D=\Sigma_B-\Sigma_A\succeq0$, adding independent $G\sim\calN(0,D)$ to an $A$ observation exactly simulates a $B$ observation. The randomized rule $\delta_B(X_A+G)$ therefore transfers every bounded-loss risk from $B$ to $A$, so optimal Bayes and minimax risks under $A$ are no larger \cite{Blackwell1953}. This permits singular covariances and requires no approximation. For full location families $\Theta=\R^d$, Loewner order is also necessary \cite{Torgersen1991}. The ordering concerns optimal risks over decision rules, including randomized rules. Comparing two fixed plug-in rules is a separate problem.

\paragraph{Comparison of crossing events.}
Slepian's inequality \cite{Slepian1962} instead states that centered Gaussian vectors with equal coordinate variances and $\Cov(Z_{A,i},Z_{A,j})\le\Cov(Z_{B,i},Z_{B,j})$ for $i\ne j$ satisfy
\[
\Prob[\max_iZ_{B,i}\le t]\ge\Prob[\max_iZ_{A,i}\le t]
\]
for every common threshold $t$. This equal-threshold comparison concerns maxima. Loewner comparisons and risk orderings for other polyhedral decisions require their own hypotheses.

\subsection{Joint coefficients and shot allocation}
\label{app:joint_allocation}
For independent scale estimators with per-shot gap variance $s_{\ell,i}^2$, actual variance is $\sum_\ell c_\ell^2s_{\ell,i}^2/B_\ell$. With $w_\ell=B_\ell/B>0$, its per-unit counterpart is $\sum_\ell c_\ell^2s_{\ell,i}^2/w_\ell$. The quadratic-over-linear terms are jointly convex in $(c_\ell,w_\ell)$, so the certified-margin design has a convex continuous allocation relaxation under Richardson constraints and $\sum_\ell w_\ell=1$. Integer rounding requires evaluating the resulting finite-budget correction.

\section{Shot-level converse, Gaussian testing and local Fisher form}
\label{app:converse_proofs}

The three parts address distinct questions: a lower bound for arbitrary processing of the shot record, exact testing in a Gaussian experiment, and a local quadratic benchmark derived from Fisher information. The exponents refer to their respective experiments and assumptions.

\subsection{Shot-level converse for fixed allocation and arbitrary post-processing}

Fix primitive settings $s\in\mathcal T$ and a non-adaptive shot allocation
$\{B_s\}_{s\in\mathcal T}$ with $\sum_{s\in\mathcal T}B_s=B$.  Let $\Theta_{\mathrm{phys}}$ denote
the model class over which the primitive per-shot laws $p_{s,\theta}$ are defined.
For an ideal-landscape instance $\theta\in\Theta_{\mathrm{phys}}$, let
\begin{equation}
    g_i(\theta)=F_\theta(x_i)-F_\theta(x_0),
    \qquad i=1,\ldots,d,
\end{equation}
be the decision-relevant target gap map, and let
\begin{equation}
    P^{(B)}_\theta=\bigotimes_{s\in\mathcal T}p_{s,\theta}^{\otimes B_s}
\end{equation}
be the law of the full shot record.  A deterministic decision rule $\delta_B$ is any
measurable map from this record to $\{0,1,\ldots,d\}$.  A randomized decision
rule is understood as a Markov kernel from the shot record to
$\{0,1,\ldots,d\}$, or equivalently as a measurable map of the record together
with an auxiliary seed $U$ that is independent of the shot data and has the same
law under every instance.

\begin{theorem}[Fixed-allocation shot-level decision converse]
\label{thm:shot_level_converse}
Assume the true instance $\nu\in\Theta_{\mathrm{phys}}$ has unique argmin $x_0$,
so $g_i(\nu)>0$ for every competitor $i$.  For each competitor $i$, define the
loss-of-optimality confuser class
\begin{equation}
    \Theta_i^{\mathrm{loss}}
    =\{\theta'\in\Theta_{\mathrm{phys}}:\, g_i(\theta')\le0\}.
\end{equation}
Thus the phrase ``physically realizable'' means realizable within the primitive
model class $\Theta_{\mathrm{phys}}$, rather than within an unspecified mitigation
protocol.  Competitors with empty confuser class are omitted from the supremum
and minimum below; if all confuser classes are empty, the statement is vacuous.
Then for every decision rule $\delta_B$,
\begin{equation}
\label{eq:shot_converse}
    \sup_i\ \sup_{\theta'\in\Theta_i^{\mathrm{loss}}}
    \max\!\left\{
    \Prob_{\nu}[\delta_B\ne0],
    \Prob_{\theta'}[\delta_B=0]
    \right\}
    \ge
    \frac14\exp\!\left[
    -\min_i\inf_{\theta'\in\Theta_i^{\mathrm{loss}}}
    \KL\!\left(P^{(B)}_\nu\,\middle\|\,P^{(B)}_{\theta'}\right)
    \right].
\end{equation}
For the following asymptotic assertion, additionally assume that
$\KL(p_{s,\nu}\|p_{s,\theta'})<\infty$ for every primitive setting $s$ and every fixed
confuser $\theta'$ in the displayed classes. If $B_s/B\to w_s$ with $w_s\ge0$ and
$\sum_{s\in\mathcal T}w_s=1$, no rule using this
fixed-allocation shot record can have a worst-pair loss-of-optimality exponent
larger than
\begin{equation}
\label{eq:shot_exponent}
    J_{\mathrm{shot}}(w)
    =\min_i\inf_{\theta'\in\Theta_i^{\mathrm{loss}}}
    \sum_{s\in\mathcal T}w_s\,
    \KL\!\left(p_{s,\nu}\,\middle\|\,p_{s,\theta'}\right).
\end{equation}
Under the additional finite-KL hypothesis, settings with $w_s=0$ contribute zero
to this per-$B$ KL rate for each fixed confuser. They may still carry information
at slower scales; deleting them does not give an equivalent statistical experiment.
Without that hypothesis, a vanishing allocation fraction does not justify
discarding a setting. The finite-budget bound remains valid with infinite KL;
an infinite product-law divergence gives a zero lower bound for that pair.
\end{theorem}

\begin{proof}
For one confuser, the two errors sum to at least $1-\TV(P_\nu^{(B)},P_{\theta'}^{(B)})$. Randomization can be represented by a common independent seed, which changes neither total variation nor KL. Bretagnolle--Huber gives $1-\TV(P,Q)\ge\tfrac12e^{-\KL(P\|Q)}$ \cite{BretagnolleHuber1979,Tsybakov2009}; taking half the error sum gives the factor $1/4$. Independence of primitive records yields
\[
\KL(P_\nu^{(B)}\|P_{\theta'}^{(B)})
=\sum_s B_s\KL(p_{s,\nu}\|p_{s,\theta'}).
\]
Optimizing over confusers proves the finite-budget claim. For the asymptotic claim, keep each confuser fixed. Finite primitive KL and a finite setting set give convergence of this sum divided by $B$ to $\sum_s w_s\KL(p_{s,\nu}\|p_{s,\theta'})$. The limsup of the worst-pair exponent is bounded by each such value; taking their infimum proves Eq.~\eqref{eq:shot_exponent}. No exchange with a moving minimizer is used.
\end{proof}

\begin{example}[A vanishing allocation fraction can retain decisive information]
Take two instances with target gaps $+1$ and $-1$. One primitive setting has the
same Bernoulli$(1/2)$ law under both. A second setting returns $0$ surely at the
first instance and $1$ surely at the second. Allocate $B-1$ observations to the
first setting and one to the second. The limiting fractions are $(1,0)$, yet
the single second-setting observation separates the instances with zero error
at every budget. Its KL divergence is infinite. Dropping it would incorrectly
assign zero information rate to a perfectly discriminating experiment. This
example explains the extra finite-KL hypothesis in the asymptotic statement.
\end{example}

\subsection{Exact Gaussian pairwise testing and pointwise argmin risk}
Consider the equal-covariance Gaussian gap experiment
\begin{equation}\label{eq:fixed_budget_gap_obs}
P=\mathcal N(\mu,\Sigma/B),\qquad Q=\mathcal N(m,\Sigma/B),\qquad \Sigma\succ0.
\end{equation}
The margin $\mu$ has positive coordinates. For a fixed loss-of-optimality alternative $m$, the binary errors are $P[\delta\ne0]$ and $Q[\delta=0]$. This binary experiment does not require identification of a named global competitor under $Q$.

\begin{corollary}[Exact Gaussian two-point testing bound]
\label{cor:twopoint_argmin_converse}
For $d_{\mu,m}^2=(\mu-m)^\top\Sigma^{-1}(\mu-m)$,
\begin{equation}\label{eq:twopoint_pair_bound}
\inf_\delta\max\{P[\delta\ne0],Q[\delta=0]\}
=\Phi\!\left(-\frac{\sqrt B\,d_{\mu,m}}2\right).
\end{equation}
The likelihood-ratio test attains equal errors and exponent $d_{\mu,m}^2/8$. The closest alternative satisfying $m_i\le0$ is
\begin{equation}\label{eq:closest_single_competitor}
m^{(i),\star}=\mu-\frac{\mu_i}{\Sigma_{ii}}\Sigma e_i,
\qquad d_{\mu,m^{(i),\star}}^2=\frac{\mu_i^2}{\Sigma_{ii}}.
\end{equation}
Its pairwise testing exponent is therefore $\mu_i^2/(8\Sigma_{ii})$. For $Q_i=\calN(m^{(i),\star},\Sigma/B)$, every rule satisfies
\begin{equation}\label{eq:twopoint_minimax_argmin}
\max_i\max\{P[\delta\ne0],Q_i[\delta=0]\}
\ge\max_i\Phi\!\left(-\frac{\sqrt B\,\mu_i}{2\sqrt{\Sigma_{ii}}}\right).
\end{equation}
Pairwise optimal tests need not form one jointly optimal rule. For singular $\Sigma$, restrict alternatives to $\mu+\operatorname{range}(\Sigma)$ and use $\Sigma^\dagger$; a coordinate with $\Sigma_{ii}=0$ cannot cross the boundary on this support.
\end{corollary}

\begin{proof}
The log likelihood ratio is Gaussian with means $\pm Bd_{\mu,m}^2/2$ and variance $Bd_{\mu,m}^2$. Thresholding at zero gives the displayed equal errors, attaining both equal-prior Bayes and minimax risk. Gaussian tail asymptotics give the rate. For $u=\mu-m$ with $u_i\ge\mu_i$, Cauchy--Schwarz gives $u^\top\Sigma^{-1}u\ge u_i^2/\Sigma_{ii}\ge\mu_i^2/\Sigma_{ii}$, with equality at $u=(\mu_i/\Sigma_{ii})\Sigma e_i$. Taking the maximum of the individual pairwise bounds proves the composite inequality; restriction to the Gaussian support proves the singular version.
\end{proof}

Here $\KL(P\|Q)=Bd_{\mu,m}^2/2$, so the shot-record inequality gives the weaker bound $\tfrac14e^{-Bd_{\mu,m}^2/2}$. The KL bound and the exact Gaussian testing risk have different exponents.

\paragraph{Named competitors and composite decisions.}
Making a specified $x_i$ the global minimizer adds constraints $m_j-m_i\ge0$ for all $j\ne i$ to $m_i\le0$. The diagonal closest-confuser formula applies only if its minimizer in Eq.~\eqref{eq:closest_single_competitor} satisfies those constraints. Otherwise additional faces bind, with the geometry of Theorem~\ref{thm:diagonal_offdiagonal}. Joint optimality for this composite task requires a separate argument.

\subsection{Local Fisher interpretation}
\label{app:local_fisher}
\begin{corollary}[Fisher quadratic benchmark]
\label{cor:shot_to_gaussian}
Fix an allocation $w$ and write
\[
I=\overline{\mathcal I}(w,\theta)=\sum_{s\in\mathcal T}w_s\mathcal I_s(\theta)\succeq0.
\]
Let $g(\theta)$ be the target gap map, $G=Dg(\theta)$, and $\mu=g(\theta)>0$. These are target gaps, not method-specific mitigated margins. Suppose
\begin{equation}\label{eq:fisher_gap_identifiable}
Gh=0\qquad\text{for every }h\in\ker I.
\end{equation}
Define the tangent covariance and its diagonal entries by
\begin{equation}\label{eq:fisher_gap_covariance}
\Sigma_{\rm eff}=GI^\dagger G^\top,\qquad s_i^2=(\Sigma_{\rm eff})_{ii}.
\end{equation}
Then the unrestricted quadratic tangent problem has the exact solution
\begin{equation}\label{eq:fisher_linear_cost}
J_{\rm lin}(w;\mu,G)
:=\inf_{h:\,\min_i\{\mu_i+(Gh)_i\}\le0}\frac12h^\top Ih
=\min_{i:s_i^2>0}\frac{\mu_i^2}{2s_i^2}.
\end{equation}
A minimum over no eligible coordinates is $+\infty$. For a row $G_i$ with $s_i^2>0$, a minimizing displacement for its crossing is
\[
h_i^\star=-\frac{\mu_i}{s_i^2}I^\dagger G_i^\top.
\]
\end{corollary}

\begin{proof}
Condition~\eqref{eq:fisher_gap_identifiable} puts $G_i^\top$ in $\operatorname{range}(I)$. Cauchy--Schwarz in the Fisher seminorm gives $(G_i h)^2\le s_i^2 h^\top Ih$. A crossing requires $G_i h\le-\mu_i$, so its cost is at least $\mu_i^2/(2s_i^2)$, with equality at $h_i^\star$. If $s_i^2=0$, the row is zero and a positive linearized margin cannot cross. Taking the infimum over the finite union proves the claim.
\end{proof}

\paragraph{Statistical interpretation and attainability.}
In a regular locally asymptotically normal (LAN) experiment on an identifiable parameterization, $\Sigma_{\rm eff}$ is the efficient covariance for the first-order target gaps if an efficient regular estimator exists. The quadratic identity supplies the covariance benchmark. If an actual regular estimator has covariance $\Sigma\succeq\Sigma_{\rm eff}$, its Gaussian plug-in rate for the same positive margins can be smaller than that based on $\Sigma_{\rm eff}$; this compares plug-in rates at the given margins. Minimax optimality additionally concerns optimization over rules and alternatives.

In the associated Gaussian experiment, the quadratic KL cost is four times the exact closest-pair testing exponent in Corollary~\ref{cor:twopoint_argmin_converse}. Attained minimax rates for the original experiment require a separate statistical analysis. Restricting tangent directions can only increase the quadratic cost. Comparing it with the nonlinear physical infimum requires the additional assumptions below.

\paragraph{Sufficient conditions for a physical local expansion.}
We next connect the tangent benchmark to the physical KL minimum in a shrinking neighborhood. Fix a competitor $i$, a sequence $\theta_\epsilon$, and $\mu_\epsilon=g_i(\theta_\epsilon)\downarrow0$. Define the per-record divergence
\[
D_w(\theta,\theta')=\sum_s w_s\KL(p_{s,\theta}\|p_{s,\theta'}).
\]
Suppose balls $\{\theta_\epsilon+h:\|h\|_2\le\rho_\epsilon\}$ are contained in $\Theta_{\rm phys}$, with $\rho_\epsilon\to0$ and $\mu_\epsilon=o(\rho_\epsilon)$. Write $I_\epsilon=\overline{\mathcal I}(w,\theta_\epsilon)$ and $q_\epsilon=Dg_i(\theta_\epsilon)$. Assume, for constants $0<c<C<\infty$ and $M<\infty$ independent of $\epsilon$,
\[
c\,\mathrm{Id}\preceq I_\epsilon\preceq C\,\mathrm{Id},
\qquad c\le\|q_\epsilon\|_2\le C,
\]
and the following uniform bounds on those balls:
\begin{align*}
\left|D_w(\theta_\epsilon,\theta_\epsilon+h)-\tfrac12 h^\top I_\epsilon h\right|
&\le\omega_\epsilon\|h\|_2^2,\qquad \omega_\epsilon\to0,\\
\left|g_i(\theta_\epsilon+h)-\mu_\epsilon-q_\epsilon h\right|
&\le M\|h\|_2^2.
\end{align*}
Then, with $s_{i,\epsilon}^2=q_\epsilon I_\epsilon^{-1}q_\epsilon^\top$,
\begin{equation}\label{eq:fisher_physical_regular}
\inf_{\substack{\|h\|_2\le\rho_\epsilon\\g_i(\theta_\epsilon+h)\le0}}
D_w(\theta_\epsilon,\theta_\epsilon+h)
=\frac{\mu_\epsilon^2}{2s_{i,\epsilon}^2}+o(\mu_\epsilon^2).
\end{equation}

\begin{proof}
Put $v_\epsilon=I_\epsilon^{-1}q_\epsilon^\top/s_{i,\epsilon}^2$, so $q_\epsilon v_\epsilon=1$ and $v_\epsilon$ is uniformly bounded. For a sufficiently large fixed $K$, $h=-(\mu_\epsilon+K\mu_\epsilon^2)v_\epsilon$ is feasible by the gap remainder bound and lies in the ball. Its divergence is $\mu_\epsilon^2/(2s_{i,\epsilon}^2)+o(\mu_\epsilon^2)$, proving the upper bound. Uniform positive definiteness and the KL remainder give $D_w\ge(c/4)\|h\|_2^2$ for small $\epsilon$. Hence near-minimizers within $o(\mu_\epsilon^2)$ of the infimum have $\|h\|_2=O(\mu_\epsilon)$. Feasibility implies $q_\epsilon h\le-\mu_\epsilon+O(\mu_\epsilon^2)$. Cauchy--Schwarz now gives the matching quadratic lower bound, and the KL remainder is $o(\mu_\epsilon^2)$ uniformly on these near-minimizers.
\end{proof}

For finitely many crossing boundaries satisfying these conditions in the same neighborhood, take the minimum of their costs. Equality with the global physical infimum additionally requires excluding cheaper confusers outside that neighborhood. A singular parameterization can be treated this way only after constructing regular coordinates that preserve the experiment, the target gaps and all relevant alternatives; annihilating $\ker I$ at one point is insufficient. A pointwise KL expansion or differentiability in quadratic mean (DQM) alone does not supply the uniform bounds and localization used in the proof. The small-margin limit above is distinct from a fixed-gap large-shot exponent; the $B^{-1/2}$ LAN scaling must also be specified separately.

\paragraph{A physical example with an exact KL minimum.}
Consider independent $Z$ measurements of two qubit states $\rho_j=(\mathrm{Id}+f_j Z)/2$, with unknown $(f_0,f_1)\in(-1,1)^2$ and true ordering $f_0<f_1$, after known common depolarization $\mathcal D_\eta(\rho)=\eta\rho+(1-\eta)\mathrm{Id}/2$, $0<\eta\le1$. Fix positive allocation fractions $w_0+w_1=1$. The outcome $+1$ has probability $p_j=(1+\eta f_j)/2$. For parameter $(f_0,f_1)$, Fisher information per total shot and the unbiased linear correction are
\[
I_w=\diag\!\left(\frac{w_j\eta^2}{1-\eta^2 f_j^2}\right)_{j=0,1},
\qquad \widehat f_j=\overline X_j/\eta,
\qquad B\Cov(\widehat f)=I_w^{-1}.
\]
The covariance identity is exact when $B_j=w_jB$; it holds in the limit when $B_j/B\to w_j$. Thus, for $G=(-1,1)$ and $\mu=f_1-f_0$, the gap coefficient is $s^2=GI_w^{-1}G^\top=\sum_j(1-\eta^2 f_j^2)/(w_j\eta^2)$.

A physical confuser has $f'_1\le f'_0$, equivalently $q_1\le q_0$ with $q_j=(1+\eta f'_j)/2\in((1-\eta)/2,(1+\eta)/2)$. Writing $\bar p=w_0p_0+w_1p_1$, its exact global KL cost within this family is
\[
J:=\inf_{q_1\le q_0}\sum_{j=0}^1w_j\kl(p_j\|q_j)
=\sum_{j=0}^1w_j\kl(p_j\|\bar p)
=\frac{\mu^2}{2s^2}+O(\mu^3).
\]
For the exact equality, each divergence is strictly convex in $q_j$, with derivative $(q_j-p_j)/[q_j(1-q_j)]$. At $q_0=q_1=\bar p$, the weighted derivatives are $(\lambda,-\lambda)$ with $\lambda=w_0w_1(p_1-p_0)/[\bar p(1-\bar p)]>0$, satisfying the optimality conditions for $q_1-q_0\le0$. The point is physical because it lies strictly between $p_0$ and $p_1$. For the last expansion, let $\mu\to0$ in a compact interior parameter set with $\eta,w$ fixed. Taylor expansion gives $J=w_0w_1(p_1-p_0)^2/[2\bar p(1-\bar p)]+O(\mu^3)$, which agrees with the displayed quadratic cost to that order.

Under known noise and fixed measurements, this example attains the covariance bound and identifies a physical KL minimizer. The quadratic expansion describes the small-margin regime; the exact fixed-gap failure exponent is a separate quantity. The common positive correction $1/\eta$ preserves empirical argmin and ties for every record, so covariance attainment here coexists with unchanged decisions. Efficiency for fitted calibration and the QAOA estimators requires analysis of those procedures.

\paragraph{Why first-order estimability is not enough.}
Let one primitive record consist of two independent Bernoulli variables with probabilities $(1+\theta_1)/2$ and $(1+\theta_2^3)/2$, on a small open neighborhood of the origin. Set $g(\theta)=\theta_1-\theta_2^2$ and $\theta_\epsilon=(\epsilon,0)$. Then
\[
\mu=\epsilon,\qquad I=\diag((1-\epsilon^2)^{-1},0),
\qquad G=(1,0),\qquad \Sigma_{\rm eff}=1-\epsilon^2.
\]
Condition~\eqref{eq:fisher_gap_identifiable} holds, and the linear boundary direction leading to $(0,0)$ is feasible. Nevertheless, $(\epsilon,\sqrt\epsilon)$ is another physical boundary alternative, with
\[
D(\theta_\epsilon,(\epsilon,\sqrt\epsilon))
=-\tfrac12\log(1-\epsilon^3)\sim\tfrac12\epsilon^3,
\qquad J_{\rm lin}=\frac{\epsilon^2}{2(1-\epsilon^2)}\sim\tfrac12\epsilon^2.
\]
The model is identifiable, smooth and DQM, yet its singular Fisher information misses a cheaper nonlinear route to the boundary. First-order estimability and a feasible linear minimizer therefore do not identify the physical infimum: the uniform coercivity required for Eq.~\eqref{eq:fisher_physical_regular} fails.

If Eq.~\eqref{eq:fisher_gap_identifiable} fails, the pseudoinverse is not an efficiency covariance for the affected first-order gap functional. This alone implies neither zero exact KL distance to a fixed-gap confuser nor impossibility of exponential discrimination, as the following example shows.

\begin{example}[Fisher singularity does not preclude a fixed-gap exponent]
Let $Y\sim\operatorname{Bernoulli}(1/2+\theta^3/4)$ for $\theta\in[-1,1]$, with target gap $g(\theta)=1/2+\theta$ and true parameter $\theta=0$.
Here $\mathcal I(0)=0$ while $g'(0)=1$. The boundary confuser $\theta'=-1/2$ has success probability $q=15/32$, whereas $p=1/2$ at the truth, so $\kl(p\|q)=0.00195694966>0$.
A test of the empirical mean against $(p+q)/2$ has each error bounded by
$\exp[-B(p-q)^2/2]=\exp[-B/2048]$ by Hoeffding's inequality, uniformly over all $\theta\le-1/2$ on the alternative side.
The model is identifiable but nonregular to first order at zero; the fixed-gap testing exponent is nevertheless positive.
\end{example}

\section{Other decisions, controls and joint risk}

\subsection{Tolerant, \texorpdfstring{top-$k$}{top-k}, ranking and step decisions}
\label{app:topk_rules}
Let $\calX_\epsilon^\star=\{x:F(x)\le F^\star+\epsilon\}$ and $\widehat{\calM}_m=\argmin_x\widehat F_m(x)$. Conservative tolerant success requires $\widehat{\calM}_m\subseteq\calX_\epsilon^\star$. If outside candidates exist, failure is exactly
\begin{equation}
\min_{i\notin\calX_\epsilon^\star}\widehat F_m(i)
\le\min_{j\in\calX_\epsilon^\star}\widehat F_m(j)
\quad\Longleftrightarrow\quad
\exists i\notin\calX_\epsilon^\star\ \forall j\in\calX_\epsilon^\star:
\widehat F_m(i)-\widehat F_m(j)\le0.
\end{equation}
This is a finite union of intersections of pairwise-gap events, with kernel $L_\epsilon K_mL_\epsilon^\top$. It counts an outside--inside tie at the minimum as failure; a chosen representative with a different tie rule requires different boundary inequalities. The event is empty if every candidate is inside the target set. Restricting the contrast matrix to a boundary band is justified only when omitted contrasts are controlled.

For an ideal top-$k$ set $S_k^\star$, with ties counted against success, failure occurs if some $i\notin S_k^\star$ has $\widehat F_m(i)\le\widehat F_m(j)$ for some $j\in S_k^\star$. The rows $e_i-e_j$ therefore define its decision kernel. For ranking, the expected number of strict inversions is
\[
\E[I_m]=\sum_{(i,j):F(i)<F(j)}\Prob[\widehat F_m(i)-\widehat F_m(j)>0].
\]
Replace $>0$ by $\ge0$ if ties count as failures. Step acceptance, $\widehat F_m(x_{\rm new})-\widehat F_m(x_{\rm old})<-\eta$, is one gap threshold with its own equality convention.

\subsection{Residual controls and non-Gaussian tails}
\label{app:affine_controls}
A positive-affine map $\widehat F=\alpha F+c\1$, $\alpha>0$, preserves every gap sign and exact tie. With a fixed tie rule, rankings and selected representatives are preserved. A further structured residual $r$ acts through $Lr$: nonconstant ratios $(Lr)_i/\Delta_i$ indicate nonuniform distortion, but do not alone imply a decision change. The distortion can remain smaller than every relevant margin.

For centered, non-Gaussian finite-budget gaps $G_m^{(B)}$, coordinatewise sub-Gaussian proxy variances $v_i^{(B)}$ give, when all $\mu_{m,i}>0$,
\begin{equation}
R_m(B)\le\sum_{i:\,v_i^{(B)}>0}\exp\!\left[-\frac{\mu_{m,i}^2}{2v_i^{(B)}}\right].
\end{equation}
A zero sub-Gaussian proxy implies that the centered gap is zero almost surely; with a positive margin, that coordinate contributes no failure. An empty sum is zero. This conservative union bound ignores dependence. For positive proxies $v_i^{(B)}=s_i^2/B$, each summand has exponent $-B\mu_{m,i}^2/(2s_i^2)$. Nonpositive margins require separate treatment.

\subsection{Equal gap marginals do not determine joint risk}
\label{app:joint_witnesses}

Let $G\sim\calN(0,\Sigma_\rho)$, where $\Sigma_\rho=s^2\bigl(\begin{smallmatrix}1&\rho\\\rho&1\end{smallmatrix}\bigr)$ and $-1<\rho<1$. Lift it to $E=(C,C+G_1,C+G_2)$ with independent centered Gaussian $C$. Every one-dimensional reference-gap law, and the ambient MSE $\Var(C)+2s^2/3$, is independent of $\rho$. Nevertheless $\Prob[G_1\le-\Delta\ \text{or}\ G_2\le-\Delta]$ changes with $\rho$. Thus even exact marginal gap probabilities do not determine the joint decision risk.

\section{Hardware micro-cell: protocol and results}
\label{app:hardware_microcell}

A preregistered exploratory four-qubit micro-cell (path-graph MaxCut, QAOA depth 1, four candidates and three active gaps) ran on the 156-qubit device \texttt{ibm\_marrakesh}. The protocol used blind calibration-based qubit selection, interleaved method scheduling and cluster bootstrap by epoch. An initial path, with held-out generalization error $0.75$, did not meet the characterization criterion. The same blind rule selected a replacement under a protocol amendment fixed and recorded before decision-covariance data collection.

On the selected path, a local depolarizing model estimates rates on the specified qubits and edges, using $p=1-e^{-r}$ to map each fitted rate to its depolarizing parameter. A retrospective fit on 162 characterization records gives held-out error $0.097$ on 41 reserved records, meeting the characterization criterion. The final model is then refitted to all 203 records. The support-stability score is $1.0$ within the specified support; it does not assess recovery of an unknown support. Covariance predictions retain the mitigation coefficients used in the executed circuits. For $R=60$ independent repetitions over 10 epochs, the relative Frobenius errors between $LMC_{\rm dev}M^\top L^\top$ and the directly estimated covariance were $0.126$ for Accuracy-ZNE, $0.173$ for Raw and $0.209$ for DA-ZNE. The first two lie within the specified $0.20$ band. The predicted ZNE-to-Raw variance amplification of approximately $7.4\times$ also appeared in the measured kernels.

The decision-direction criterion requires projected ranking error below $0.25$ times the predicted risk margin at the upper 95\% cluster-bootstrap bound. The upper bounds exceed this threshold for all three pairs. The Accuracy-ZNE/Raw point estimate lies below the threshold, while the Accuracy-ZNE/DA-ZNE pair has a near-zero predicted margin. The table uses analytic conditional-Gaussian derivatives, assigning half of each symmetric off-diagonal derivative to each matrix entry, and 10,000 whole-epoch bootstrap resamples of the archived counts.

\begin{center}\small
\begin{tabular}{lrrr}\toprule
Pair & Projected error & 95\% upper bound & Required bound\\\midrule
Accuracy-ZNE / DA-ZNE & $0.0060209$ & $0.0139553$ & $0.00000009$ \\
Accuracy-ZNE / Raw & $0.0022532$ & $0.0062904$ & $0.00279613$ \\
DA-ZNE / Raw & $0.0037677$ & $0.0112278$ & $0.00279604$ \\
\bottomrule\end{tabular}\end{center}

Persistent calibration drift can add covariance beyond the shot-scaling prediction in Section~\ref{subsec:anchoring_cdev}, providing one possible source of mismatch. The present data leave decision ordering unresolved. A larger repetition budget and a prespecified margin floor for active pairs would improve the resolution of this test.

\section{Additional covariance and calibration diagnostics}
\label{app:extended_experiments}

The full gap covariance changes Gaussian risk most at low budget: the largest full-minus-diagonal difference is $-0.2865$ for P002 Accuracy-ZNE at $B=256$, and differences become negligible at high budget. The selected method is nevertheless identical in all 24 instance--budget cells, represented by 72 method rows (Fig.~\ref{fig:gate0_full_diag}). This distinguishes a change in the value of a risk estimate from a change in the method minimizing it.

\begin{figure}[!ht]
    \centering
    \includegraphics[width=0.74\linewidth]{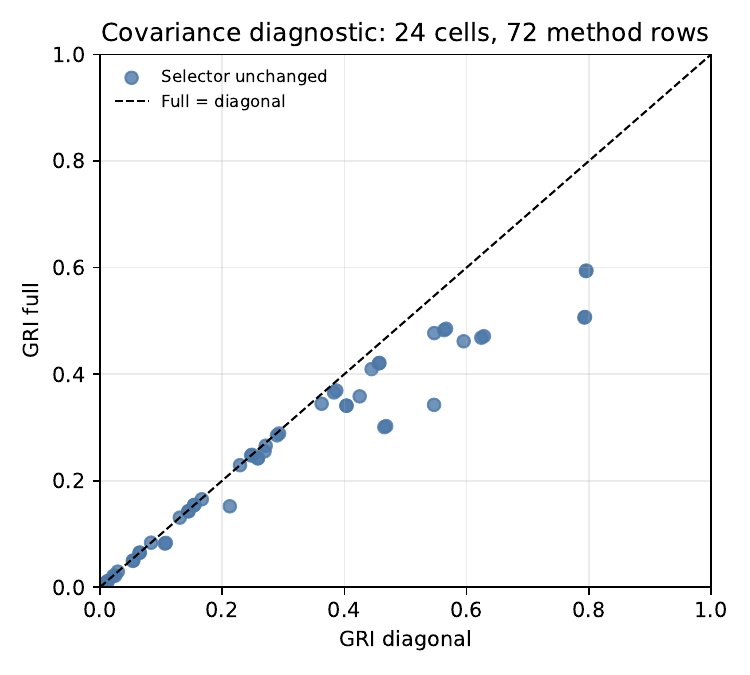}
    \caption{Effect of retaining off-diagonal gap covariance in the Gaussian risk score. The effect on score values is largest at low budget; the minimizing method is unchanged in all 24 instance--budget cells. The 72 plotted points represent the three methods in each cell.}
    \label{fig:gate0_full_diag}
\end{figure}

A preliminary sample-mean calibration check placed 70/72 cells within tolerance $\tau_R=0.05$, with two low-budget ZNE cells outside. It used the same repetitions to fit the Gaussian score and estimate failure.

Analytic primitive means agree with the shot data at $B_{\max}=8192$, with maximum $z$-scores below 2.76 across P001--P004. A learned uniform two-parameter depolarizing channel recovers the simulated global channel with relative weighted error $e_C=0.0096$. This diagnostic addresses two global rates; the hardware experiment separately fits local rates on a specified support.

\FloatBarrier
\section{Numerical methods}
\label{app:implementation_checks}

\paragraph{E1 partition sensitivity.}\label{app:e1_split_sensitivity}
Within each instance, five deterministic SHA256 orderings of the 40 replicate IDs form four groups of ten. Four cyclic assignments per ordering allocate one group to fitting, the next to selection and the remaining two to testing. Each replicate therefore appears five times in fitting, five in selection and ten in testing. These 20 partitions were fixed before their results, after the original E1 outcomes were known, and held identical across budgets. Each partition refits and reselects with the original methods, ideal references, shot blocks and costs. The original split remains the primary analysis and is excluded from the alternative-partition average.

For uncertainty, rates are first averaged over partitions within each instance, followed by 10,000 paired bootstrap resamples of the 48 complete instance clusters. At $t=512$, the average full-cost GRI-minus-Raw difference is $0.0843$, with descriptive 95\% interval $[0.0578,0.1110]$; GRI minus MSE selection is $-0.1535$, with interval $[-0.1956,-0.1119]$. At equal evaluation cost, GRI minus Raw is $0.0067$, with interval $[-0.0037,0.0157]$; that interval also includes zero at $t=1024$. Reusing records measures sensitivity to their partition without adding independent observations or testing transfer to new landscapes.

\paragraph{Shot-block reconstruction.}
Here $t$ is the prefix length from Section~\ref{sec:e1_reanalysis}. Archived cumulative means recover the observed shot blocks exactly:
\[
\overline Y_{(a,b]}=\frac{b\overline Y_b-a\overline Y_a}{b-a}.
\]

For the full-cost comparison, the Raw baseline uses $(t/2,8t]$ at scale one. Fixed Accuracy-ZNE combines disjoint blocks $(t/2,t]$ and $(2t,4t]$ at each scale, giving $2.5t$ shots per scale. Fixed Decision-ZNE uses $2t$ per scale. These blocks lie within the archived prefixes.

\paragraph{Risk calculation and exact decisions.}
The coefficient solver normalizes each covariance form by its squared positive margin and then applies one common scale before optimization. It retains zero eigenvalues without an absolute variance floor. Primal and dual residuals and the duality gap determine numerical acceptance. Gaussian scores use standardized coordinates and stable tail calculations. General covariance derivatives use Cholesky conditioning and an analytic bivariate density; the supported double-precision regime requires correlation condition number at most $10^{14}$. For gaps sharing one incumbent, the implementation instead evaluates the full analytic gradient by one-dimensional conditional-normal integration. Symmetric off-diagonal derivatives contribute half to each matrix entry, so Frobenius contraction gives the first-order risk change.

For integer-valued MaxCut records, empirical ZNE decisions reconstruct shot sums and compare rational weighted sums exactly. Accuracy-ZNE uses weights $(4,1,-2)/3$. Fitted three-scale Richardson coefficients are stored through one parameter $\tau$, with $c=(2+\tau,-1-2\tau,\tau)$; rational evaluation of that binary parameter preserves the affine constraints. For a strict single-selection decision, any nonpositive gap, including an exact tie, counts as failure. In the dynamic search, exact estimator ties retain the first candidate within a batch and the earlier incumbent across rounds. Its success criterion is ideal regret at most $\epsilon$ relative to the best candidate visited along that trajectory; $\epsilon$ is fixed from the archived ideal landscape.

The 48 fixed-block device-noise instances use fixed streams for local probability calculations and sampling.

\subsection{Signed-PEC control}
\label{app:signed_pec}

\paragraph{Design and cost.}
Eight random six-vertex 3-regular graphs use depth-two QAOA with nine candidates each: a randomly perturbed nominal angle vector and symmetric offsets along its four coordinates. Centers are perturbed uniformly by up to $0.20$ radians and offsets have radius $0.12$. Candidate order is randomized; ideal values are not used to construct or select candidates. The primitive gates remain $H$, $R_X$ and $R_{ZZ}$. A $k$-qubit depolarizing channel $\mathcal E_p(\rho)=(1-p)\rho+pI/2^k$ follows each gate, with $(p_1,p_2)=(0.005,0.015)$ or $(0.01,0.03)$. PEC uses the exact Pauli decomposition of its inverse and ideal correction gates. Each candidate and repetition uses independent records, with a new correction branch for every shot. Both methods receive $B\in\{768,3072,12288\}$ shots per candidate, hence $9B$ per decision; $\Gamma^2$ is a variance overhead, not a further shot charge. The channel is stipulated known.

\paragraph{Signed law and covariance.}
For gate-level inverse coefficients $q_{\ell j}$, let $\gamma_\ell=\sum_j|q_{\ell j}|$, sample branch $j$ with probability $|q_{\ell j}|/\gamma_\ell$, and put $\Gamma=\prod_\ell\gamma_\ell$. Let $s\in\{-1,+1\}$ be the product of coefficient signs. The subnormalized positive and negative output sectors obey $\rho_++\rho_-=\rho_{\rm abs}$ and $\rho_+-\rho_-=\rho_{\rm ideal}/\Gamma$. Here $\rho_{\rm abs}$ is the normalized circuit output with the absolute-weight correction mixture. For a measurement outcome $z$, this gives
\begin{equation}\label{eq:pec_signed_law}
\Prob(s,z)=\tfrac12\bigl(P_{\rm abs}(z)+sP_{\rm ideal}(z)/\Gamma\bigr).
\end{equation}
The simulation samples this exact marginal over circuit branches. Ideal probabilities enter its construction and evaluation, not an estimator fitted to ideal targets. The marginal law was checked against explicit branch enumeration and 2,048 independently sampled correction circuits with one shot each.

For $m$ graph edges, write the minimized objective as $-\operatorname{Cut}(z)=-m/2+h(z)$, where $h(z)=m/2-\operatorname{Cut}(z)$. The known constant is added exactly, so one PEC contribution and its variance are
\begin{align}
Y&=-m/2+\Gamma s h(z), & \E Y&=F(x),\label{eq:pec_centered_estimator}\\
v_x=\Var(Y)&=\Gamma^2\E_{\rm abs}[h^2]-(\E_{\rm ideal}h)^2.\label{eq:pec_signed_variance}
\end{align}
Each candidate mean has variance $v_x/B$, giving $K_{\PEC}^{(B)}=\diag(v_x)/B$ and gap covariance $LK_{\PEC}^{(B)}L^\top$. This retains correlations between gaps sharing the reference candidate. In general $v_x\ne\Gamma^2\Var_{\rm ideal}(h)$. The two noise levels give $\Gamma^2=3.6105$ and $13.0750$.

\paragraph{Finite-shot evaluation.}
There are 512 independent repetitions per landscape, noise level, budget and method. Failure means that the unique ideal minimizer is not the unique estimated minimizer; ties count as failures. Since $\Gamma$ is constant across candidates within each cell, PEC decisions maximize the integer sum of $T=\operatorname{Cut}(z)$ for $s=+1$ and $T=m-\operatorname{Cut}(z)$ for $s=-1$. Raw maximizes its cut sum. A $B$-fold convolution of each integer statistic, evaluated by a zero-padded FFT, gives the finite-shot decision probability including ties. Gaussian risk uses the actual means and covariance; across the 48 landscape--noise--budget cells its maximum absolute error is $0.00348$ for PEC and $0.00321$ for Raw.

The primary comparison averages both noises equally within each landscape and uses 10,000 paired bootstrap resamples of the eight landscapes, keeping the two noise levels together. At $B=3072$, the empirical MSE difference (PEC minus Raw) is $-0.03144$, with exploratory 90\% interval $[-0.04744,-0.01637]$. Its sign and the failure difference in Table~\ref{tab:signed_pec} agree with the finite-shot model. At $B=768$, the MSE-difference interval includes zero; the aggregate pattern is therefore not uniformly resolved across budgets. These small-sample intervals describe this controlled comparison. Risk is evaluated with known moments, rather than learned from finite calibration data.

\paragraph{Relation to the effective-shot proxy.}
The earlier control in Fig.~\ref{fig:accuracy_decision_separation} uses $n_{\rm eff}=\lfloor B/\gamma^2\rfloor$ ideal multinomial shots per candidate, giving covariance $\Sigma_{\rm ideal}/n_{\rm eff}$ for ideal single-shot covariance $\Sigma_{\rm ideal}$. It models overhead without signed contributions. The aggregate inequalities $\MSE_{\PEC}<\MSE_{\Raw}$ and $\GRI_{\PEC}>\GRI_{\Raw}$ hold at budgets $768,1536,3072,6144$, and jointly on $5/8,7/8,7/8,7/8$ instances, respectively. POM\_004 never satisfies both: Raw has low MSE at small budgets, whereas the proxy becomes decision-better at high budgets. Those evaluation budgets exclude learning a physical correction. The signed control uses a separate candidate set and its actual observable-dependent variance, providing a complementary test of the same accuracy--decision separation.

\end{document}